\documentclass[nonacm]{acmart}

\renewcommand\footnotetextcopyrightpermission[1]{}
 
\setcopyright{none}

\usepackage{booktabs}   
\usepackage{subcaption} 
\usepackage{multicol}
\usepackage{wrapfig}
\usepackage{tikz}
\usetikzlibrary{external}
\usepackage{tikz-cd}
\usepackage{tikzit}
\usepackage{stackengine}
\usepackage{enumitem}
\usepackage{stmaryrd}
\usepackage{listings}
\usepackage{marvosym} 
\usepackage{amsmath,amscd}

\tikzstyle{sd*smallbox}=[fill=white, draw=black, shape=rectangle]
\tikzstyle{sd*box}=[fill=white, draw=black, shape=rectangle, minimum height=0.75cm, minimum width=0.75cm]
\tikzstyle{sd*tallbox}=[fill=white, draw=black, shape=rectangle, minimum height=1.25cm, minimum width=0.75cm]
\tikzstyle{sd*widebox}=[fill=white, draw=black, shape=rectangle, minimum width=1.25cm]
\tikzstyle{sd*2widebox}=[fill=white, draw=black, shape=rectangle, minimum height=0.75cm, minimum width=1.75cm]
\tikzstyle{sd*3widebox}=[fill=white, draw=black, shape=rectangle, minimum height=0.75cm, minimum width=2.25cm]
\tikzstyle{hg*smallbox}=[fill=white, draw=black, shape=rectangle, rounded corners, tikzit shape=circle]
\tikzstyle{hg*box}=[fill=white, draw=black, shape=rectangle, rounded corners, minimum height=0.75cm, minimum width=0.75cm, tikzit shape=circle]
\tikzstyle{hg*tallbox}=[fill=white, draw=black, shape=rectangle, rounded corners, minimum height=1.25cm, minimum width=0.75cm, tikzit shape=circle]
\tikzstyle{sd*dot}=[none, tikzit fill={rgb,255: red,128; green,128; blue,128}, tikzit shape=rectangle]
\tikzstyle{hg*dot}=[fill=black, shape=circle, inner sep=0pt, minimum height=0.1cm, minimum width=0.1cm]
\tikzstyle{sd*dup}=[fill=black, shape=circle, scale=0.5]
\tikzstyle{hg*dup}=[fill=white, draw=black, inner sep=2mm, shape=rectangle, rounded corners, tikzit shape=circle, label={center:$\Delta$}]
\tikzstyle{sd*tupler}=[none, inner sep=3pt, fill=white, draw=black, shape=circle, thick, tikzit shape=rectangle]
\tikzstyle{sd*eval}=[fill=white, draw=black, thick, shape=semicircle, rotate=180, inner sep=0pt, minimum height=0.35cm, minimum width=0.35cm, tikzit shape=rectangle]
\tikzstyle{sd*unit}=[fill=white, draw=black, thick, shape=semicircle, inner sep=0pt, minimum height=0.35cm, minimum width=0.35cm, tikzit shape=rectangle]
\tikzstyle{hg*eval}=[fill=white, draw=black, shape=rectangle, rounded corners=0.5mm, minimum height=0.35cm, minimum width=0.35cm, label={center:$ev$}, tikzit shape=circle]
\tikzstyle{sd*wtrap}=[fill=white, draw=black, thick, shape=trapezium, inner sep=0.1pt, minimum height=0.3cm, minimum width=0.3cm, tikzit shape=rectangle]
\tikzstyle{hg*wtrap}=[fill=white, draw=black, thick, shape=trapezium, inner sep=0.1pt, minimum height=0.3cm, minimum width=0.3cm]
\tikzstyle{sd*btrap}=[fill=black, draw=black, thick, shape=trapezium, inner sep=0.1pt, minimum height=0.3cm, minimum width=0.3cm, tikzit fill=black, tikzit shape=rectangle]
\tikzstyle{hg*btrap}=[fill=black, draw=black, thick, shape=trapezium, inner sep=0.1pt, minimum height=0.3cm, minimum width=0.3cm, tikzit fill=black]
\tikzstyle{sd*wtrap'}=[fill=white, draw=black, thick, shape=trapezium, rotate=180, inner sep=0.1pt, minimum height=0.3cm, minimum width=0.3cm, tikzit shape=rectangle]
\tikzstyle{hg*wtrap'}=[fill=white, draw=black, thick, shape=trapezium, rotate=180, inner sep=0.1pt, minimum height=0.3cm, minimum width=0.3cm]
\tikzstyle{sd*discard}=[fill=black, shape=circle, scale=0.5]
\tikzstyle{hg*discard}=[fill=white, draw=black, shape=rectangle, rounded corners=0.5mm, minimum height=0.35cm, minimum width=0.35cm, label={center:$!$}, tikzit shape=circle]
\tikzstyle{sd*state}=[fill=white, draw=black, thick, isosceles triangle, isosceles triangle apex angle=90, rotate=90, scale=0.8, inner sep=0.01pt, minimum height=0.3cm, minimum width=0.15cm, tikzit shape=rectangle]
\tikzstyle{hg*state}=[fill=white, draw=black, thick, isosceles triangle, isosceles triangle apex angle=90, rotate=90, scale=1.2, inner sep=0.01pt, minimum height=0.3cm, minimum width=0.3cm]
\tikzstyle{functor}=[text={rgb,255: red,0; green,128; blue,128}, tikzit draw={rgb,255: red,0; green,128; blue,128}]

\tikzstyle{sd*curve}=[-, tikzit draw=blue]
\tikzstyle{sd*thickcurve}=[-, very thick, tikzit draw={rgb,255: red,0; green,128; blue,128}]
\tikzstyle{hg*curve}=[->, >=latex]
\tikzstyle{sd*morph}=[->, fill=none, draw=none, tikzit draw=orange]
\tikzstyle{hg*morph}=[->, tikzit draw=orange]
\tikzstyle{sd*ad}=[-, fill={rgb,255: red,195; green,236; blue,255}, draw={rgb,255: red,0; green,87; blue,134}, double]
\tikzstyle{sd*bg}=[-, fill={rgb,255: red,245; green,237; blue,229}, draw=none, tikzit draw=Apricot, tikzit fill={rgb,255: red,215; green,207; blue,200}]
\tikzstyle{sd*interface}=[-, fill=none, draw=none]
\tikzstyle{sd*fwd}=[-, fill={rgb,255: red,244; green,178; blue,121}, draw={rgb,255: red,141; green,61; blue,0}, double]
\tikzstyle{sd*rev}=[-, double, fill={rgb,255: red,105; green,240; blue,202}, draw={rgb,255: red,0; green,104; blue,73}]
\tikzstyle{sd*bubble}=[-, fill=white, draw=black, thick, rounded corners=5mm]
\tikzstyle{hg*ad}=[-, fill={rgb,255: red,103; green,103; blue,255}, draw=blue, rounded corners=1mm]
\tikzstyle{hg*bg}=[-, fill={rgb,255: red,191; green,191; blue,191}, draw=gray, dashed]
\tikzstyle{hg*interface}=[-, fill={rgb,255: red,211; green,211; blue,255}, draw=blue, dashed]
\tikzstyle{hg*fwd}=[-, fill={rgb,255: red,255; green,13; blue,94}, draw={rgb,255: red,191; green,0; blue,64}, rounded corners=1mm]
\tikzstyle{hg*rev}=[-, fill={rgb,255: red,177; green,255; blue,177}, draw={rgb,255: red,0; green,103; blue,0}, rounded corners=1mm]
\tikzstyle{hg*bubble}=[-, fill=white, draw=black, thick, rounded corners=5mm]
\tikzstyle{sd*grad}=[-, fill={rgb,255: red,244; green,172; blue,213}, draw={rgb,255: red,113; green,67; blue,93}, double]
\tikzstyle{invisible}=[-, fill=none, draw=none, tikzit draw=black, dashed]
\tikzstyle{greenfunctor}=[-, fill=none, draw={rgb,255: red,0; green,128; blue,128}, tikzit draw={rgb,255: red,0; green,128; blue,128}]

\usepackage{xparse}
\NewDocumentCommand{\Hyp}{O{} O{}}{\mathsf{Hyp}_{\scriptscriptstyle #1}^{\scriptscriptstyle #2}} 

\NewDocumentCommand{\HyperG}{O{} O{}}{\mathcal{H}_{\scriptscriptstyle #1}^{\scriptscriptstyle #2}} 

\def\F{{\mathcal F}}
\def\G{{\mathcal G}}
\def\H{{\mathcal H}}
\def\K{{\mathcal K}}
\def\L{{\mathcal L}}

\def\R{{\mathcal R}}

\def\Set{{\mathbf{Set}}}

\def\id{\text{id}}

\newcommand{\diagram}[1]{
    \tikzpicturedependsonfile{#1.tikz}
    \tikzsetnextfilename{#1}
    \input{#1.tikz}
}
\newcommand{\hd}[1]{\scalebox{.5}{\diagram{#1}}}

\newcommand{\idty}{\hd{pics/tikzit/components/id-typed}}

\newcommand{\seqty}{\hd{pics/tikzit/components/sequential-composition-typed}}

\newcommand{\morphism}{\hd{pics/tikzit/components/morphism-typed}}

\newcommand{\functorcomp}{\diagram{pics/tikzit/pam/functor-comp}}
\newcommand{\functorid}{\diagram{pics/tikzit/pam/functor-id}}

\newcommand{\constzero}{\scalebox{.5}{\hspace{-10pt}\diagram{pics/tikzit/components/zero-sd}}}
\newcommand{\constbarzero}{\scalebox{.5}{\hspace{-10pt}\diagram{pics/tikzit/components/barzero-sd}}}
\newcommand{\opgeneric}[2]{\scalebox{.5}{
    \tikzpicturedependsonfile{pics/tikzit/components/op-generic.tikz}
    \tikzsetnextfilename{generic/#2}
    \renewcommand{\GENERIC}[0]{#1}
    \hspace{-10pt}
    \input{pics/tikzit/components/op-generic.tikz}
}}
\newcommand{\gradgeneric}[2]{\scalebox{.5}{
    \tikzpicturedependsonfile{pics/tikzit/components/op-generic-grad.tikz}
    \tikzsetnextfilename{generic/#2}
    \renewcommand{\GENERIC}[0]{#1}
    \hspace{-10pt}
    \input{pics/tikzit/components/op-generic-grad.tikz}
}}
\newcommand{\opadd}{\scalebox{.5}{\hspace{-10pt}\diagram{pics/tikzit/components/add-sd}}}
\newcommand{\opoplus}{\scalebox{.5}{\hspace{-10pt}\diagram{pics/tikzit/components/oplus-sd}}}

\newcommand{\GENERIC}[0]{}

\newcommand{\binomsd}[2]{\left(\stackanchor{$#1$}{$#2$} \right)}
\newcommand{\seval}[1]{\llbracket #1 \rrbracket}
\newcommand{\boldit}[1]{\textbf{\textit{#1}}}

\definecolor{fwd-color}{RGB}{244,178,121}
\definecolor{rev-color}{RGB}{105,240,202}
\definecolor{adj-color}{RGB}{195,236,255}
\definecolor{grad-color}{RGB}{244,172,213}

\newtheorem{remark}{Remark}

\begin{document}

\title[Functorial string diagrams for RAD]
{Functorial String Diagrams for Reverse-Mode Automatic Differentiation}

\author{Mario Alvarez-Picallo}
\affiliation{
  \department{Programming Languages Laboratory}
  \institution{Huawei Research Centre}
  \streetaddress{2 Semple Str.}
  \city{Edinburgh}
  \state{Scotland}
  \postcode{EH3 8BL}
  \country{United Kingdom}
}
\email{mario.alvarez.picallo@huawei.com}

\author{Dan R. Ghica}
\affiliation{
  \position{Director}
  \department{Programming Languages Laboratory}
  \institution{Huawei Research Centre}
  \streetaddress{2 Semple Str.}
  \city{Edinburgh}
  \state{Scotland}
  \postcode{EH3 8BL}
  \country{United Kingdom}
}
\email{dan.ghica@gmail.com}
\affiliation{
  \department{Computer Science}
  \institution{University of Birmingham}
  \city{Birmingham}
  \state{England}
  \postcode{B15 2TT}
  \country{United Kingdom}
}

\author{David Sprunger}
\affiliation{
  \department{Computer Science}
  \institution{University of Birmingham}
  \city{Birmingham}
  \state{England}
  \postcode{B15 2TT}
  \country{United Kingdom}
}
\email{d.sprunger@bham.ac.uk}

\author{Fabio Zanasi}
\affiliation{
  \department{Computer Science}
  \institution{University College London}
  \city{London}
  \state{England}
  \postcode{WC1E 6BT}
  \country{United Kingdom}
}
\email{f.zanasi@ucl.ac.uk}

\begin{abstract}
We enhance the calculus of string diagrams for monoidal categories with hierarchical features in order to capture closed monoidal (and cartesian closed) structure. 
Using this new syntax we formulate an automatic differentiation algorithm for (applied) simply typed lambda calculus in the style of~\cite{DBLP:journals/toplas/PearlmutterS08} and we prove for the first time its soundness. 
To give an efficient yet principled implementation of the AD algorithm we define a sound and complete representation of hierarchical string diagrams as a class of hierarchical hypergraphs we call \emph{hypernets}. 
\end{abstract}



\keywords{string diagrams, automatic differentiation, graph rewriting, hierarchical hypergraphs}  

\maketitle
\thispagestyle{empty}

\section{Introduction}

This paper covers three main topics which support, motivate, and reinforce each other: 
\boldit{reverse automatic differentiation (AD)}, \boldit{string diagrams}, and \boldit{(hyper)graph rewriting}. 

AD is an established technique for evaluating the derivative of a function specified by a computer program, a particularly challenging exercise when the program contains higher-order sub-terms. 
This technique came to recent prominence due to its important role in algorithms for machine learning~\cite{DBLP:journals/jmlr/BaydinPRS17}.
We focus in particular on the influential algorithm defined in~\cite{DBLP:journals/toplas/PearlmutterS08}, which lies at the foundation of many practical implementations of AD. 
The main novel contribution of our paper is to prove, for the first time, the \boldit{soundness} of this particular style of AD algorithm. 

String diagrams are a formal graphical syntax used in the representation of morphisms in monoidal categories~\cite{selinger2010survey} which is finding an increasing number of applications in a wide range of mathematical, physical, and engineering domains. 
We contribute to the development of string diagrams by formulating a new \textbf{\textit{hierarchical string diagram calculus}}, with associated equations, suitable for the representation of closed monoidal (and cartesian closed) structures. 
This innovation is, as we shall see in the paper, warranted:  the hierarchical string diagrammatic syntax allows for a more intelligible formulation of a complex algorithm and, most importantly, a new style of inductive argument which leads to a relatively simple proof of soundness.

Finally, hierarchical hypergraphs are given as a concrete and efficient representation of hierarchical string diagrams, which pave the way towards efficient and effective implementation of AD as graph rewriting in the well established framework of \boldit{double-pushout (DPO) rewriting}~\cite{DBLP:conf/mfcs/EhrigK76}. 
Moreover, we identify a class of {hierarchical hypergraphs}, which we call \boldit{hypernets}, which are a \boldit{sound and complete} representation of the hierarchical string diagram calculus. 
This is the third and final contribution of our paper.


\section{Higher-order string diagrams}
  
String diagrams are a convenient alternative notation for constructing morphisms, in particular in (strict) monoidal categories. 
In this paper we largely build on the syntax proposed in~\cite{DBLP:conf/csl/Mellies06}, with only a few cosmetic changes aimed at making higher-order concepts more perspicuous. 
String diagrams in this work are to be read from top to bottom.

\subsection{Functorial string diagrams}

We start with the basic language of categories, ranged over by $\mathcal C, \mathcal D\ldots$.
This language consists of a collection of \emph{objects} ranged over by $A, B, \ldots$ 
and two families of terminal symbols, \emph{identities} $\id_A:A\to A$, represented as an $A$-labelled vertical stem~\idty
, and morphisms $f:A\to B$, represented by labelled boxes with 
an $A$-labelled top stem (which we sometimes call \emph{input} or \emph{operand}) and a $B$-labelled bottom stem (which we call \emph{output} or \emph{result})~\morphism. 
We may distinguish (families of) terminal symbols in the diagram language with particular geometrical shapes instead of labelled boxes,
much in the way we have artificially disgtinguished identities from other morphisms. 

Terms, ranged over by $f, g, \ldots$, are created using composition. 
Given  $f:A\to B$, and $g:B\to C$ we write $f;g:A\to C$ as the \emph{stacking} of the diagrams for $f$ and $g$. 
Since the output of $f$ must match the input of $g$ we connect the corresponding stems, to give a graph-like appearance to the string diagram~\seqty. 
We enforce two properties of composition familiar from category theory.
First, composition is associative, meaning $(f;g);h=f;(g;h)=f;g;h$.
This identification is subsumed by the diagrammatic notation.
Second, we require the identity axiom $f;\id=\id;f=f$.
Diagrammatically, this means the lengths of the stems of a diagram can be lengthened or shortened without ambiguity.

\begin{wrapfigure}{r}{0.2\textwidth}
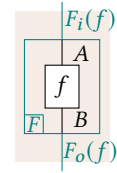

\centering\diagram{pics/tikzit/pam/combinatorbox}
\caption{An example frame}
\label{fig:frame}
\end{wrapfigure}

We extend our string diagram language with \emph{labelled frames} 
which indicate mappings between morphisms of different categories. 
The application of a mapping $F$ to a morphism, as a string diagram, is indicated by 
drawing an $F$-labelled frame around the morphism and 
modifying the stems of the diagram as appropriate, as seen in Fig.~\ref{fig:frame}.
Note that in this diagram the stems and morphims inside the frame are
a different color to the frame and the stems outside the frame.
This is an indication that objects and morphisms belong to potentially distinct categories.
When the map goes from a category to itself, we may use the same color inside and out of the frame,
but often leave the frame itself a different color to emphasize the mapping.

Such morphism mapping $F$ constitutes a \emph{functor} if it satisfies the following properties.
First, there must be a mapping on objects, 
which we also denote $F$ by common abuse of notation,
such that $F_i(f) = F(A)$ and $F_o(f) = F(B)$ for all $f: A \to B$ in the source category.
Second, this mapping must respect basic categorical structures,
expressed in the language of string diagrams in Fig.~\ref{fig:functor-properties}.
We use $1_{\mathcal C}$ for the identity functor. 
Given two functors $F, G$ with matching domains and codomains we write $FA=F(A)$ and $FGA=F(G(A))$ in the sequel. 

\begin{figure}[h]
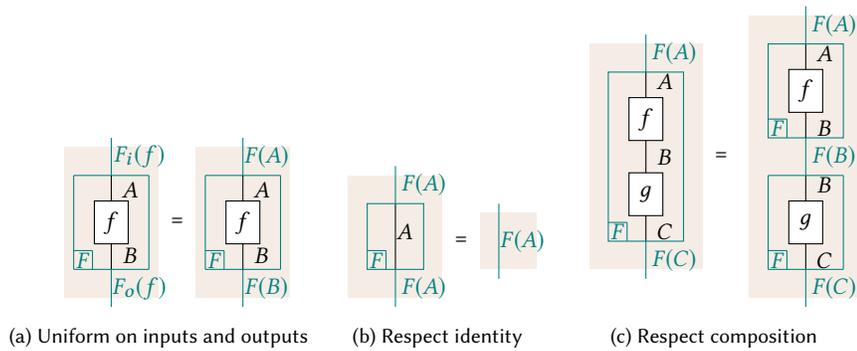

\centering
\subcaptionbox{Uniform on inputs and outputs}
{\quad\diagram{pics/tikzit/pam/functordef}\quad}
\subcaptionbox{Respect identity}{\functorid}
\subcaptionbox{Respect composition}{\functorcomp}
\caption{Properties of functors}
\label{fig:functor-properties}
\end{figure}

\newpage
The diagrammatic notation can be generalised to bifunctors in the obvious way, by drawing two side-by-side boxes for the arguments. 
One bifunctor that plays a special role in string diagram is the \emph{tensor product} or \emph{monoidal product}, in particular when it is \emph{strict}.  
The tensor is represented diagrammatically as:

\begin{center}
\diagram{pics/tikzit/pam/bifunctor} = 
\diagram{pics/tikzit/pam/parallel-comp-strict} = 
\diagram{pics/tikzit/pam/parallel-comp}
\end{center}

The diagram above contains three representations. 
In the first one we can see tensor as a bifunctor, with the two separate boxes indicating the two arguments of the bifunctor. 
The second one is special notation for the tensor, essentially hiding the functorial box and using a graphical convention (the horizontal line) to represent the `unravelling' of the tensored-labelled stem into components. 
Finally, the third one is special notation for strict monoidal tensor, in which the tensor $A\otimes C$ is represented as the list of its components $[A, C]$. 
The strict diagram absorbs the associativity isomorphisms and makes the tensor associative \emph{on the nose}: 

\begin{center}
\diagram{pics/tikzit/pam/left-assoc} =
\diagram{pics/tikzit/pam/triple-assoc} = 
\diagram{pics/tikzit/pam/right-assoc}
\end{center}

Henceforth we will work in the strict setting, but it will be sometimes convenient to \emph{de-strictify} a diagram and group individual stems in stems with tensor types. 
Coherence (and strictness) ensure that such de-strictifiations can be always performed unambiguously. 

In the strict setting we also have special notation for the unit $I$ of the tensor, which we represent as the empty list; identity on $I$ is represented as empty space. 
It is immediate then, diagrammatically, that $f\otimes \id_I=\id_I\otimes f=f$. 

We further extend the string diagram with the concept of \emph{natural transformation} between functors with the same domains and same codomains. 
Natural transformations are object-indexed families of morphisms written as $\theta_A:FA\to GA$ (or just $\theta:F\to G$) which obey the following family of axioms, expressed in the language of string diagrams as:

\begin{center}
\diagram{pics/tikzit/pam/natl-trans-pre} = 
\diagram{pics/tikzit/pam/natl-trans-post}
\end{center}

One particularly interesting example of natural transformation is \emph{symmetry}, 
written as $\gamma_{A,B}:A\otimes B\to B\otimes A$, 
for which we use the special geometric shape of two crossing wires. 
The fact that it is a natural transformation immediately imples that

\begin{center}
\diagram{pics/tikzit/pam/swap-def}\quad 
\diagram{pics/tikzit/equations/natural-swap}
\end{center}

Symmetry is also an involution, i.e. $\gamma_{A,B};\gamma_{B,A}=\id_A\otimes \id_B$.

For functors $F:\mathcal D\to \mathcal C$ and $G:\mathcal C\to \mathcal D$ such that natural transformations $\epsilon:FG\to 1_{\mathcal C}$ (called \emph{the counit}), $\eta:1_{\mathcal D}\to GF$ (called \emph{the unit}) exist, they form an adjunction if and only if they satisfy the following family of axioms:

\begin{center}
For all 
\diagram{pics/tikzit/pam/counit} and 
\diagram{pics/tikzit/pam/unit} we have that 
\diagram{pics/tikzit/pam/unitcounit2} and 
\diagram{pics/tikzit/pam/unitcounit1}.
\end{center}
In this situation, we say $F$ is a left adjoint and $G$ is the right adjoint.

We adopt the convention of writing the counit of an adjunction as a downward pointing semicircle,
the unit as an upward pointing semicircle, and 
omitting the label when the map is clear from context. 
Note that~\cite{DBLP:conf/csl/Mellies06} does not discuss adjunctions specifically, 
although the streamling of the notations and calculations with adjunctions is a prime benefit of the string diagram notation, and 
no additional technical content is required. 

\subsection{String diagrams for monoidal-closed and cartesian-closed categories}

Monoidal closed categories and cartesian closed categories are categorical models for the linear and simply-typed lambda calculus, respectively. 
A monoidal closed category arises if for every object $X$ in the category, 
the (endo)functor $F_X(A)= A\otimes X$ has a right adjoint $G_X(A)=X\multimap A$.
Diagrammatically, we depict these functors as:

\begin{center}
\diagram{pics/tikzit/pam/right-product}\quad and \quad
\diagram{pics/tikzit/pam/exponential}
\end{center}

Instantiated to the families of functors $F_X, G_X$ above, 
the naturality and adjunction equations are expressed in string diagrams as in Fig.~\ref{fig:mcc}. 
The counit of the adjunction is normally called \textit{eval}, and 
we call the unit \textit{coeval} for the sake of symmetry in terminology and by analogy with compact-closed categories.
\begin{figure}[ht]
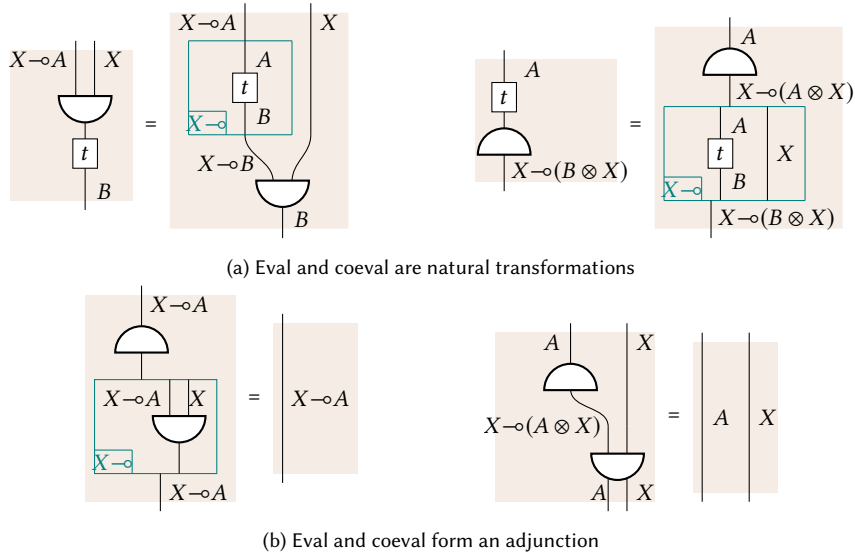

\centering
\subcaptionbox{Eval and coeval are natural transformations}[\textwidth]{
\diagram{pics/tikzit/equations/natural-counit}\qquad\qquad
\diagram{pics/tikzit/equations/natural-unit}}
\subcaptionbox{Eval and coeval form an adjunction}[\textwidth]{
\diagram{pics/tikzit/pam/coevaleval}\qquad\qquad
\diagram{pics/tikzit/pam/evalcoeval}}
\caption{String diagram representation of MCC axioms.}
\label{fig:mcc}
\end{figure}

To further expand our diagrammatic language to cartesian closed categories, one easy way is to add natural transformations $\delta_A:A\to A\otimes A$ (contraction) and $\omega_A:A\to I$ (weakening) such that 
$\delta_A;\omega_A\otimes \id_A = \id_A=\delta_A;\omega_A\otimes \id_A$~\cite{heunen2012lectures}. 
We represent both of these natural transformations with a black dot, disambiguated by the quantity of results.
The monoid equations are~\scalebox{.6}{\diagram{pics/tikzit/equations/copy-monoid}}. 
Copying and discarding are both consequences of naturality, 
i.e.~\scalebox{.5}{\diagram{pics/tikzit/equations/natural-copy}} 
and~\scalebox{.5}{\diagram{pics/tikzit/equations/natural-discard}}, respectively. 

Here we have presented adjunctions with unit and counit natural transformations.
An equivalent description of adjunctions involves a natural bijection between sets of morphisms.
In the case of monoidal or cartesian closed categories, this bijection is between
${\mathcal C}(\Gamma\otimes A, B)$ and ${\mathcal C}(\Gamma, A \multimap B)$.
This bijection is known as ``currying'', and is a more germane presentation for the lambda calculus.
We define \emph{abstraction}, 
the composition of the unit of the adjunction with the functorial box for $G$, 
as syntactic sugar denoted by a plain box with rounded corners. 
\begin{center}
\diagram{pics/tikzit/ex-sub/pam-abstraction} $:=$
\diagram{pics/tikzit/ex-sub/abstractionab}
\end{center}
This structure for abstraction gives our notion of a \emph{hierarchical string diagram},
which is to say a string diagram which may contain other string diagrams in these boxes. 

\subsubsection{Foliations}

Terms written as string diagrams can be presented in a particular form, which will turn out to lead to some useful insights:
\begin{definition}[Foliations]\label{def:foliation}
A \textit{foliation} is a string diagram written as the sequential composition of a list of diagrams called \textit{leafs}. 
A \textit{singleton leaf} is a diagram consisting of a non-identity atomic string diagram (symmetry, evaluation, operation, contraction, or weakening) or an abstraction, tensored with any number of identities. 
A \textit{maximally sequential foliation} is a foliation comprising only singleton leafs. 
A \textit{maximally sequential hierarchical foliation} is a maximally sequential foliation which is either abstraction free, or in which all abstracted diagrams are also maximally sequential hierarchical foliations. 
\end{definition}
For instance, if $f,g$ are not identities then the maximally sequential foliations of $f\otimes g$ \scalebox{.5}{\diagram{pics/tikzit/components/foliation-fxg}} 
are \scalebox{.5}{\diagram{pics/tikzit/components/foliation-fg}} $(f\otimes id);(id\otimes g)$ and \scalebox{.5}{\diagram{pics/tikzit/components/foliation-gf}} $(id\otimes g);(f\otimes id)$. 
The following is an obvious generalisation of a folklore theorem about monoidal categories.
\begin{lemma}\label{lem:max-seq-foliation}
Any hierarchical string diagram can be written as a (non-unique) maximally sequential hierarchical foliation. 
\end{lemma}
The proof is straightforward.  
The graphical intuition which underlies the proof is that whenever two morphisms are ``level'' in a diagram one of them can be ``shifted'' using identities, then tensors and compositions can be reorganised using the functorialty of the tensor. 

Foliations are convenient because syntactic transformations can be presented recursively on the foliation. 
This spares us the need to define `big' rules for sequential and tensorial composition. 
Instead only `small' rules for composing a term with a singleton leaf are required. 
This makes transformations easier to specify, and also makes for simpler inductive proofs, using the foliation as a list.

\subsection{Explicit substitution in string diagrams}

In this section we illustrate the use of hierarchical string diagrams to represent the simply typed lambda calculus with explicit substitutions. 
This is an interesting example in its own right, but more importantly it sets the scene for the next section, where we define an automatic differentiation algorithm.
The explicit substitutions play an essential role, as they give us a handle on managing closures, which the AD algorithm requires. 

Hierarchical string diagrams with rules for copying and discarding are a ready-made graphical syntax for the lambda calculus with explicit substitutions~\cite{DBLP:journals/jfp/AbadiCCL91}.
Syntactically, these calculi fall mainly in two categories, those using deBruijn indices or those using named variables. 
The former have better formal properties and their formalisation can be mechanised, but are not a very human-readable notation. 
The latter are easier to read but have some subtle failures of alpha equivalence. 
Formalising alpha equivalence for calculi of explicit substitution is a somewhat tricky problem, the solution of which leads back to rather intricate notations~\cite{DBLP:journals/iandc/FernandezG07}.
String diagrams thus seem like an improved syntax for explicit substitutions, as they are both formal and, we contend, rather readable. 
The graphical notation is variable-free therefore alpha equivalence is not an issue, and other equational properties are also rendered obvious by the diagrammatic representation. 

We pick for comparison a presentation of the calculus of explicit substitutions with named variables~\cite{DBLP:conf/csl/Kesner07}, leaving aside alpha equivalence. 
A $\Lambda$es-term is inductively defined as a variable $x$, an application $t\ u$, an abstraction
$\lambda x.t$ or a \textit{substituted term} $t[x/u]$, where $t$ and $u$ are $\Lambda$es-terms and $x$ a variable. 
The terms $\lambda x.t$ and $t[x/u]$ both bind $x$ in $t$. 
The set of free variables of a term
$t$, denoted $\overline t$ is defined as usual. 

Note that the syntactic object
$[x/u]$, an \textit{explicit substitution},  is not a term because of the way variable $x$ is bound. 
By contrast, in deBruijn formulations of the lambda calculus with explicit substitutions, substitutions are terms. 

The following key equations and reduction rules are considered:
\begin{align}
	t[x/u][y/v] &= t[y/v][x/u] &y\not\in \overline u \land x\not\in \overline v \tag{CE}\\
	(\lambda x.t) u &\rightarrow t[x/u]  \tag{BR} \\
	x[x/u] &\rightarrow u \tag{Var} \\
	t[x/u] &\rightarrow t & x\not\in \overline t \tag{Gc}\\
	(t\ u)[x/v] &\rightarrow t[x/v]\ u[x/v] \tag{App} \\
	(\lambda y.t)[x/v] &\rightarrow \lambda y.t[x/v] & x\neq y \tag{Lamb} \\
	t[x/u][y/v] &\rightarrow t[x/u[y/v]] \tag{Comp}
\end{align}
We qualify these as `\textit{key}' because for more precise resource analysis the (Lamb) and (Comp) rewrites usually are given `linear' forms depending on whether the substituted variable occurs in the term or not. 
For our purposes here we can assume the linear versions are subsumed by the general version and the (Gc) axiom.

\begin{figure}
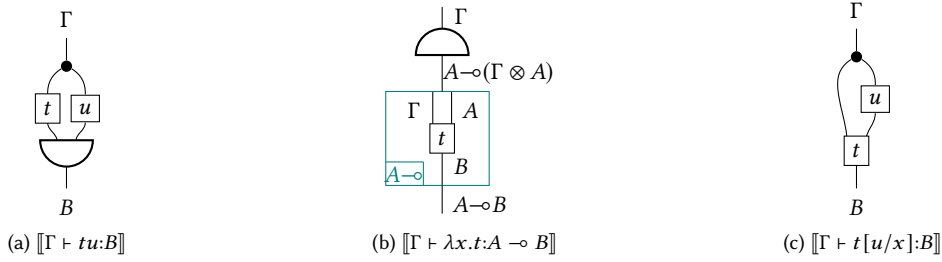

\centering
\subcaptionbox{$\seval{\Gamma \vdash tu{:} B}$}[.3\textwidth]{
\diagram{pics/tikzit/ex-sub/application}}
\hfill
\subcaptionbox{$\seval{\Gamma \vdash \lambda x.t{:} A \multimap B}$}[.3\textwidth]{
\diagram{pics/tikzit/ex-sub/pam-abstraction}}
\hfill
\subcaptionbox{$\seval{\Gamma \vdash t[u/x] {:} B}$}[.3\textwidth]{
\diagram{pics/tikzit/ex-sub/substitution-reflected}}
  \caption{Interpretation of $\Lambda$es-terms}
  \label{fig:interpret}
\end{figure}

The string diagram interpretation of $\Lambda$es is given in Fig.~\ref{fig:interpret}.
The proof of soundness is a straightforward exercise. 
Some of the axioms are simply instances of the  identity (Var), associativity of composition (Comp), and naturality of the contraction (App) or weakening (Gc) --- we leave them as an exercise. 
The two non-trivial axioms and their proofs are in Fig.~\ref{fig:lesmodel}.

\begin{figure}
\centering
\subcaptionbox{Proving (BR) using functoriality of $-\otimes A$, naturality of eval, and adjunction cancellation}{
\qquad\qquad\qquad
\diagram{pics/tikzit/pam/br1} =
\diagram{pics/tikzit/pam/br2} =
\diagram{pics/tikzit/ex-sub/substitution-reflected}
\qquad\qquad\qquad
}
\\[1.5ex]
\subcaptionbox{Proving (Lamb) using naturality of coeval and functoriality of $A\multimap-$}{
\qquad
\diagram{pics/tikzit/pam/lamb1} =
\diagram{pics/tikzit/pam/lamb2} =
\diagram{pics/tikzit/pam/lamb3}
\qquad
}
\caption{Axioms of the $\Lambda$es and their proofs.}
\label{fig:lesmodel}
\end{figure}

Finally, the CE structural rule is also rather interesting, as it requires proving that~\mbox{\scalebox{.4}{\diagram{pics/tikzit/ex-sub/ce-lhs}} = \scalebox{.4}{\diagram{pics/tikzit/ex-sub/ce-rhs}}}.
The proof is an immediate consequence of the functoriality of the tensor and of the identity law. 
What is interesting is that the two diagrams look very similar \emph{as graphs}. 
Indeed, the intuition that diagrams represented by isomorphic graphs denote equal morphisms will be made rigorous in Sec.~\ref{sec:hypernet}. 

To emphasise the syntactic nature of the transformation we will call the objects the \emph{types} of the diagrams. 
Since we are situated in a strict-monoidal setting we will write a composite tensor of objects as a list of types.
We write a generic typed term $t:A_1\otimes\cdots\otimes A_m\multimap A_1'\otimes\cdots\otimes A_n'$ in the language of string diagrams as \opgeneric{t}{t} $:\binomsd{[A_1,\ldots,A_m]}{[A_1',\ldots,A_n']}$.

\section{A graphical AD algorithm}
\label{sec:ad}

This section represents the main technical result of our paper, to define and and prove the soundness of an algorithm for performing reverse-mode automatic differentiation on hierarchical string diagrams. 
The algorithm can be considered a simplified version of that presented in~\cite{DBLP:journals/toplas/PearlmutterS08}. 
This algorithm is remarkable for being one of the first such algorithms that can be applied to code containing closures and higher-order functions. 
It is particularly in the treatment of higher-order features where we draw inspiration from their work. 

The soundness property of the algorithm is technically interesting because simple inductive proofs of correctness do not seem possible.
If simply taking the gradient of a higher-order function, the algorithm is actually unsound. 
However, when taking the gradient of a function with ground-type inputs and outputs only, the results are correct even if the function contains higher-order terms. 

Unlike the original algorithm, however, we \emph{do not provide automatic differentiation as a first-class entity}. 
This means, implicitly, that we also do not have a means to perform `\textit{higher order differentiation}' in the sense of differentiating the differential operator itself. 
In the original work, this was achieved by extending the language with rich runtime reflection capabilities whose formalisation is entirely outside of the scope of our paper. 
Our algorithm instead is formulated as a meta-level set of rules on hierarchical string diagrams or, in actual implementation, on their hypernet representation. 
This is akin to the source-to-source transformation approach to automatic differentiation.

The setting for this algorithm is that of a (strict) cartesian closed category generated from one object $o::=\mathcal{R}$, representing the real numbers, and a collection of primitive operations (addition, multiplication, trigonometric functions, etc.) and their gradients, along with a collection of nullary primitive operations for real constants. 
Among these, real addition \opadd and zero \constzero must be included.
In the string diagrams throughout this section we represent constants as a triangle instead of a box, just for improved readability.
We write \opoplus and \constbarzero for the obvious extension of \opadd and \constzero to bundles.

Each of the provided primitive operations must also come equipped with a \emph{pullback diagram}: for a primitive operation \opgeneric{op}{op} of type $\binomsd{B}{B'}$, its pullback diagram \opgeneric{op^\star}{op-bp} must have type $\binomsd{B :: B'}{B}$.
We make no assumptions about the pullback diagram of an operator, other than its type. 
However, the correctness result in this section will require pullback diagrams to be `correct' implementations of the gradient of the corresponding operation.
\nopagebreak
\subsection{Rewrite rules on string diagrams}

\newcommand{\rewrite}[0]{\mapsto}
The AD algorithm consists of three separate sets of transformations, the application of which we denote by differently coloured boxes around a diagram. We emphasise that these
boxes represent \emph{meta-level transformations}, and are not to be confused with object-level entities such as the rounded rectangles that we use to denote abstraction. 
\pagebreak[3]

To reduce clutter we use coloured boxes rather than labelled frames to indicate string-diagram transformations.
The first transformation, whose only rewrite rule can be found in Fig.~\ref{fig:backprop-sd}, is denoted by a {\color{adj-color} blue} box and is the entry point of the algorithm. Given an 
input diagram with operands of type $B$ and results of type $B'$, this transformation produces an \emph{adjoint} diagram with operands of type $B$ and results of type $B :: [B' \multimap B]$, corresponding to the
result of the original diagram plus an abstraction, the \emph{backpropagator}, that computes the gradient of the original diagram at the point at which the adjoint diagram is evaluated.
In particular, if the original diagram produces a single result of type $\mathcal{R}$, when evaluating the backpropagator at $1.0$ we will obtain the gradient of the original diagram.

\begin{wrapfigure}{r}{0.35\textwidth}
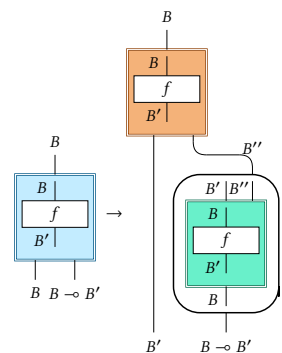

	\centering
  \scalebox{0.7}{\diagram{pics/ad/rewrite/backprop-sd}}
	\caption{Adjoint rule}
	\label{fig:backprop-sd}
\end{wrapfigure}

This transformation consists, in turn, of two components: a \emph{forward pass} transformation (in {\color{fwd-color} orange}), rewrite rules in Fig.~\ref{fig:fwd-sd}) and a \emph{reverse pass} transformation 
(in {\color{rev-color} green}), rewrite rules in Fig.~\ref{fig:rev-sd}). As the naming suggests, these correspond to the forward and reverse passes commonly employed in reverse-mode AD.
The forward pass executes the original function `as is', whereas the reverse pass computes the gradient of every sub-expression, in reverse order of execution. 
In our algorithm, as is usually the case in reverse-mode AD systems, some intermediate values computed during the forward pass are preserved and passed along
to the diagram corresponding to the reverse pass. This is shown in Fig.~\ref{fig:backprop-sd} as a bundle of type $B''$ flowing from the forward-pass computation into the backpropagator.

\newcommand{\type}[1]{\scriptscriptstyle #1}

\tikzsetnextfilename{fwd-rev-rules}
\begin{figure}
  \begin{subfigure}{\textwidth}
    \centering
    \scalebox{0.7}{\diagram{pics/ad/rewrite/fwd-const-sd}}\hfill
    \scalebox{0.7}{\diagram{pics/ad/rewrite/fwd-quiver-sd}}\hfill
    \scalebox{0.7}{\diagram{pics/ad/rewrite/fwd-copy-sd}}\hfill
    \scalebox{0.7}{\diagram{pics/ad/rewrite/fwd-discard-sd}}\vspace{2\baselineskip}
    \scalebox{0.7}{\diagram{pics/ad/rewrite/fwd-lambda-sd}}\hfill
    \scalebox{0.7}{\diagram{pics/ad/rewrite/fwd-eval-sd}}\hfill
    \scalebox{0.7}{\diagram{pics/ad/rewrite/fwd-op-sd}}
    \subcaption{Rewrites defining the forward pass}
    \label{fig:fwd-sd}
  \end{subfigure}\vspace{\baselineskip}
  \begin{subfigure}{\textwidth}
    \scalebox{0.7}{\diagram{pics/ad/rewrite/rev-const-sd}}\hfill
    \scalebox{0.7}{\diagram{pics/ad/rewrite/rev-quiver-sd}}\hfill
    \scalebox{0.7}{\diagram{pics/ad/rewrite/rev-copy-sd}}\hfill
    \scalebox{0.7}{\diagram{pics/ad/rewrite/rev-discard-sd}}\vspace{2\baselineskip}
    \scalebox{0.7}{\diagram{pics/ad/rewrite/rev-lambda-sd}}\hspace{5pt}
    \scalebox{0.7}{\diagram{pics/ad/rewrite/rev-eval-sd}}\hspace{5pt}
    \scalebox{0.7}{\diagram{pics/ad/rewrite/rev-op-sd}}
    \subcaption{Rewrites defining the reverse pass}
    \label{fig:rev-sd}
  \end{subfigure}
  \caption{Forward and reverse pass rewrites}
  \label{fig:fwd-rev-sd}
\end{figure}

The rewrites for the forward-pass transformation, depicted in Fig.~\ref{fig:fwd-sd}, are self-explanatory, as they are limited to constructing a copy of the original diagram. 
Only two cases (those for evaluation and abstraction) merit some attention. 
For abstraction, the diagram enclosed by the bubble is recursively transformed using the blue rule --- that is, any abstraction 
in the primal diagram is replaced by a new abstraction that computes the adjoint of the original one.
Then, when function evaluation in the primal diagram is translated by the forward pass, the result of the adjoint application contains both the result of the original abstraction and a backpropagator which is not used in the forward pass but
is set aside for the reverse pass.

\begin{figure}
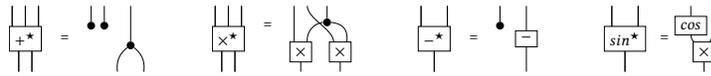

  \centering
  \scalebox{.7}{\diagram{pics/ad/rewrite/bp-plus-sd}}\hspace{15pt}
  \scalebox{.7}{\diagram{pics/ad/rewrite/bp-times-sd}}\hspace{15pt}
  \scalebox{.7}{\diagram{pics/ad/rewrite/bp-minus-sd}}\hspace{15pt}
  \scalebox{.7}{\diagram{pics/ad/rewrite/bp-sin-sd}}
  \caption{Pullback diagrams for some common operations}
  \label{fig:pullback-graphs}
\end{figure}

The rules governing the reverse pass transformation, in Fig.~\ref{fig:rev-sd}, are more involved, so we provide here an intuitive explanation for each. The first rule, which handles constants, states
that constants do not contribute to the gradient of the graph. The second rule computes the gradient of the identity
function to be the identity. 
Contraction is transformed into addition, since the gradient of the diagonal map $\langle\id, \id\rangle$ is the addition of tangent vectors, and weakened variables become zero. 
Each primitive operation is replaced by its corresponding pullback diagram, which receives as additional operands the copies of the inputs to the operation in the forward pass.
This is why we require that every primitive operation 
to be mapped to a pullback diagram of the appropriate type. Some examples of pullback diagrams for common operations can be found in Fig.~\ref{fig:pullback-graphs}.

The reverse pass handling of application and abstraction are, both in the original algorithm and in our interpretation of it, difficult to back up with compelling intuitions, but we shall try our best. 

Remember that the forward pass transforms every abstraction in the primal diagram in order to compute the original value together with a backpropagator.
The latter is captured by the reverse pass in every application rule. 
When rewriting an application node, the reverse pass instead applies the backpropagator given by the forward rule. 
This backpropagator in turn produces a wire for every operand of the body of the original abstraction,
which are swapped into the correct order.
The abstraction rule in the reverse pass then expects the sensitivity of an abstraction to consist of
a bundle of wires corresponding to the sensitivities of each captured wire.

\begin{figure}
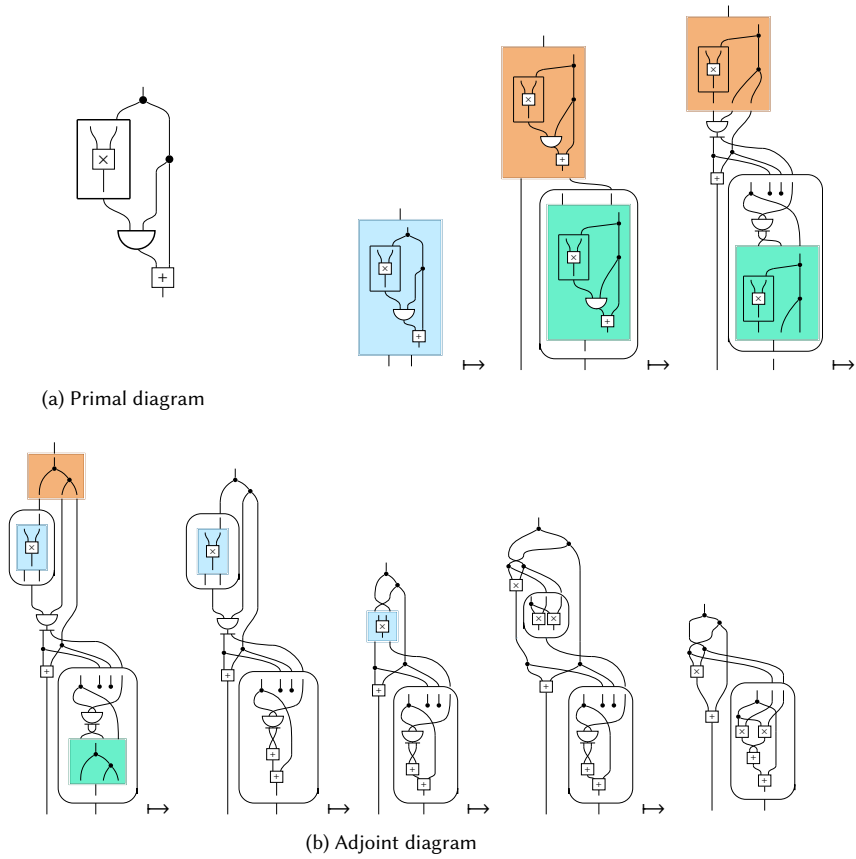

  \subcaptionbox{Primal diagram\label{fig:example-primal}}[0.35\textwidth]{
    \centering
    \scalebox{0.7}{\diagram{pics/ad/example/example-primal}}
  }
  \captionsetup[subfigure]{labelformat=empty}
  \subcaptionbox{}{
    \scalebox{0.4}{\diagram{pics/ad/example/example-1-standalone}}
    $\rewrite$
    \scalebox{0.4}{\diagram{pics/ad/example/example-2-standalone}}
    $\rewrite$
    \scalebox{0.4}{\diagram{pics/ad/example/example-4-standalone}}
    $\rewrite$
  }
  \\
  \setcounter{subfigure}{1}
  \captionsetup[subfigure]{labelformat=parens}
  \subcaptionbox{Adjoint diagram\label{fig:example-adjoint}}{
    \scalebox{0.4}{\diagram{pics/ad/example/example-5-standalone}}
    $\rewrite$
    \scalebox{0.4}{\diagram{pics/ad/example/example-6-standalone}}
    $\rewrite$
    \scalebox{0.4}{\diagram{pics/ad/example/example-7-standalone}}
    $\rewrite$
    \scalebox{0.4}{\diagram{pics/ad/example/example-8-standalone}}
    $\rewrite$
    \scalebox{0.4}{\diagram{pics/ad/example/example-10-standalone}}
  }
  \caption{The AD transformation on a string diagram diagram computing $x^2 + x$
  \label{fig:example-ad}}
\end{figure}

As an example illustrating this algorithm and its handling of closures in particular, we provide in Fig.~\ref{fig:example-primal} a diagram that might result from a program like
\lstinline[columns=fixed]{let mul y = x * y in mul x + x}, with the free variable \lstinline[columns=fixed]{x} corresponding to its single operand. On the right, in 
Fig.~\ref{fig:example-adjoint}, we show the result of applying the adjoint transformation to this diagram (see Appendix~\ref{appendix:animations} for an animated step-by-step derivation). It is 
a mere calculation to check that the resulting backpropagator, when applied to input $1$, can be evaluated to the correct derivative of the polynomial $x^2 + x$.

\begin{lemma}
  The rewriting systems in Fig.~\ref{fig:fwd-sd} and Fig.~\ref{fig:rev-sd} have the diamond property, up to a permutation of the output wires.
\end{lemma}
\begin{proof}[Proof (Sketch)]
  Every rule erases one node in the fringe of the left-hand side diagram, 
  and that no two rules can be applied to erase the same node. Therefore, if two
  rules can apply to the same diagram, it must be the case that they apply to different fringe nodes. It is then easily checked that every pair of such rules
  commutes, modulo a permutation of the wires that are propagated from the forward to the reverse pass. For a concrete example, consider the two sequences of rewrites 
  in Fig.~\ref{fig:diamond}.
\end{proof}

  \begin{figure}[b]
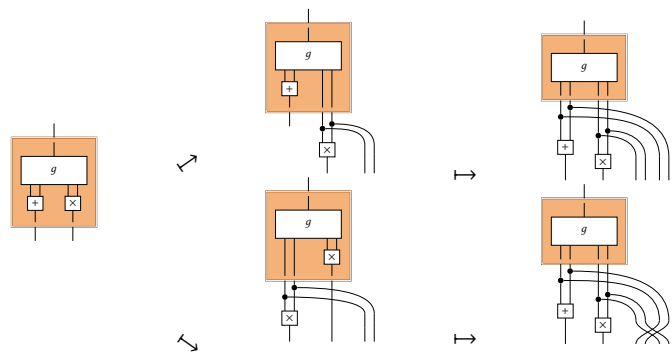

    \centering
    \begin{minipage}{0.3\textwidth}
      \hfill
      \scalebox{0.5}{\diagram{pics/ad/diamond-proof-1}}
    \end{minipage}
    \begin{minipage}{0.3\textwidth}
      \centering
      \rotatebox{35}{$\rewrite$}
      \hspace{15pt}
      \scalebox{0.5}{\diagram{pics/ad/diamond-proof-2-right}}\\
      \rotatebox{-35}{$\rewrite$}
      \hspace{15pt}
      \scalebox{0.5}{\diagram{pics/ad/diamond-proof-2-left}}
    \end{minipage}
    \begin{minipage}{0.3\textwidth}
      $\rewrite$
      \hspace{15pt}
      \scalebox{0.5}{\diagram{pics/ad/diamond-proof-3-right}}\\
      $\rewrite$
      \hspace{15pt}
      \scalebox{0.5}{\diagram{pics/ad/diamond-proof-3-left}}
    \end{minipage}
    \caption{Two possible rewrite sequences stemming from two distinct foliations\label{fig:diamond}}
  \end{figure}

\begin{remark}
The proof above, although very simple, illustrates a proof method that is made possible by using string diagrams: induction on the length of the \emph{foliation} of the diagram (Def.~\ref{def:foliation}). 
The `fringe' mentioned in the proof above is simply the `bottom' (in this case) leaf in the chosen foliation, noting that the foliation is not unique. 
This proof method also benefits additionally from absence of names and all related bureaucratic concerns (free vs. bound, alpha equivalence, capture-avoiding substitution).
\end{remark}

\subsection{Reverse derivative categories}

In order to prove that the algorithm we have given is correct, we need to select an appropriate semantic domain that reflects the behaviour of
the gradient operator from calculus. The obvious choice is the setting of reverse derivative categories \cite{cockett2019reverse}. In simple
terms, these are cartesian categories equipped with a `reverse differential combinator' which behaves, in a suitable sense, like taking the 
gradient of a function in multivariate calculus. For a more thorough treatment and explanation, we refer the reader to \emph{loc. cit.}.
They are defined as follows:

\newcommand{\<}[0]{\left\langle}
\renewcommand{\>}[0]{\right\rangle}
\renewcommand{\R}[0]{\mathsf{R}}
\begin{definition}\cite[Def.~13]{cockett2019reverse}
  A \emph{reverse derivative category} is a cartesian left-additive category endowed with a combinator $\mathsf{R}$ sending each morphism $f : X \to Y$ to
  a morphism $\mathsf{R}[f] : X \times Y \to X$ which satisfies the following conditions:
  \begin{enumerate}[label={\bf [RD.\arabic*]},align=left]
    \item $\R[f+g] = \R[f] + \R[g]$ and $\R[0]=0$
    \item $\R[u] \circ \<f,g+h \> = \R[u] \circ \< f,g \> + \R[u] \circ \<f,h\>$
      and $\R[u] \circ \<f,0\> = 0$
    \item $\R[\id ] = \pi_1$, $\R[\pi_1] = \< \pi_2, 0\>$ and $\R[\pi_2] = \<0, \pi_2\>$
    \item $\R[\<f, g\>] = \R[f] \circ (\id \times \pi_1) + \R[g] \circ (\id \times \pi_2)$
      and $\R[!] = 0$.
    \item $\R[g \circ f] = \R[f] \circ \< \pi_1, \R[g] \circ (f \times \id)\>$
    \item $\pi_2 \circ \R[\R[\R[f]]] \circ (\<\id, 0\>\times \id) \circ \<\id \times \pi_1, 0 \times \pi_2\> = \R[f] \circ (\id \times \pi_2)$
    \item $\pi_2 \circ \R[\R[\pi_2 \circ \R[\R[f]] \circ (\<0,\id\>\times \id)]] \circ (\<\id, 0\> \times \id) $\\ \mbox{}\qquad$
    = \pi_2 \circ \R[\R[\pi_2 \circ \R[\R[f]] \circ (\<0, \id\> \times \id)]] \circ (\<\id, 0\>\times \id) \circ \< \pi_1 \times \pi_1, \pi_2 \times \pi_2\>$
  \end{enumerate}
  \label{def:rdc}
\end{definition}

One caveat of reverse derivative categories is that they do not naturally accommodate higher-order functions. Indeed, there is no natural
notion of a reverse derivative category with exponentials
In contrast, cartesian differential categories which can be extended to differential $\lambda$-categories
\cite{bucciarelli2010categorical} --- cartesian differential categories which are cartesian closed and where the differential combinator is `well-behaved' with
respect to abstraction. This limitation is of no concern to us, however: we do not claim that the AD algorithm in this paper produces correct gradients for arbitrary
higher-order diagrams, only for those whose inputs and outputs have first-order types -- even if they do contain higher-order sub-terms. The first-order setting
of reverse derivative categories is sufficient.

\begin{figure}
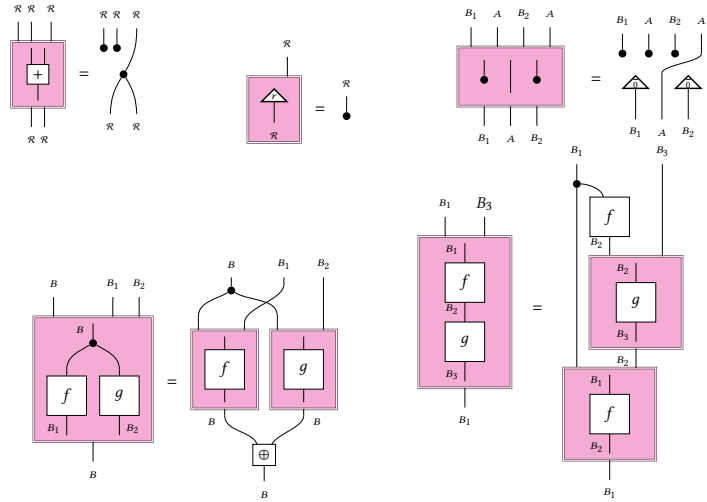

    \centering
    \scalebox{0.7}{\diagram{pics/ad/eqs/grad-add}}\hspace{30pt}
    \scalebox{0.7}{\diagram{pics/ad/eqs/grad-const}}\hspace{30pt}
    \scalebox{0.7}{\diagram{pics/ad/eqs/grad-proj}}\\
    \scalebox{0.7}{\diagram{pics/ad/eqs/grad-pair}}\hspace{20pt}
    \scalebox{0.7}{\diagram{pics/ad/eqs/grad-chain}}
  \caption{Reverse differential axioms as string diagrams}
  \label{fig:grad-sd}
\end{figure}
Henceforth, we will assume that the strict cartesian category generated by the object $\mathcal{R}$ and the collection of primitive operators and pullback
diagrams defined above is a reverse derivative category. We will use the notation \gradgeneric{f}{f-grad} to denote the 
reverse derivative $\R\big[$\opgeneric{f}{f}$\big]$ in diagrammatic form. In addition, we require that this reverse derivative category satisfies:
\begin{itemize}
  \item The $0$ of the left-additive structure coincides with \constzero
  \item The $+$ of the left-additive structure coincides with \opadd
  \item For each primitive operation \opgeneric{op}{op}, we have \opgeneric{op^\star}{op-bp} $=$ \gradgeneric{op}{op-grad}
\end{itemize}
Using this notation, all the equations in Definition~\ref{def:rdc} can be written diagramatically. The graphical translation
of conditions {\bf [RD.1]} and {\bf [RD.3]}-{\bf [RD.5]}, which will be relevant to us later, can be found in Fig.~\ref{fig:grad-sd}.

\subsection{Correctness}

Our proof of correctness proceeds in two steps. First, we prove that our AD transformation is compatible with Beta reduction, that is to say, whenever two diagrams are equivalent modulo
Beta reduction, then so are their adjoints. Then, we show that the AD transformation is correct for diagrams featuring only first-order nodes 
(that is to say, no abstractions or applications). For both of these steps, we will make use of the following technical result, which simply states that the forward and reverse passes
are compositional.

\begin{figure}[h]
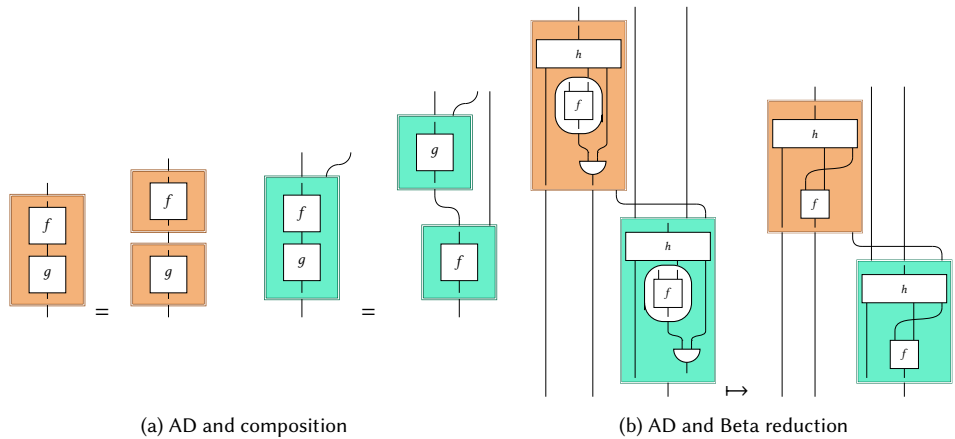

    \centering
    \subcaptionbox{AD and composition\label{eqn:ad-composition}}{
      \vspace{30pt}
      \scalebox{.65}{\diagram{pics/ad/chain-fwd-statement-lhs}}
      $=$
      \scalebox{.65}{\diagram{pics/ad/chain-fwd-statement-rhs}}
      \hspace{10pt}
      \scalebox{.65}{\diagram{pics/ad/chain-rev-statement-lhs}}
      $=$
      \scalebox{.65}{\diagram{pics/ad/chain-rev-statement-rhs}}
    }
    \subcaptionbox{AD and Beta reduction\label{eqn:ad-beta}}{
      \scalebox{.5}{\diagram{pics/ad/beta-local-statement-lhs}}
      $\rewrite$
      \scalebox{.5}{\diagram{pics/ad/beta-local-statement-rhs}}
    }
    \caption{AD is compatible with composition and reduction}
  \end{figure}

\begin{lemma}
  The forward and reverse pass rewriting rules in Fig.~\ref{fig:fwd-rev-sd} satisfy the diagrammatic version of the chain rule, that is to say, 
  Eqn.~\ref{eqn:ad-composition} holds.
\end{lemma}
\begin{proof}
  A trivial induction on the maximally sequential hierarchical foliation of $f$.
\end{proof}

\begin{lemma}
  The rewriting rules in Fig.~\ref{fig:fwd-rev-sd} are compatible with beta reduction (Eqn.~\ref{eqn:ad-beta} holds).
  \label{lem:ad-beta}
\end{lemma}
\begin{proof}
  The proof proceeds by straightforward application of the rewrite rules.
  We provide the calculation in full in Fig.~\ref{fig:ad-beta-proof} (an animated version of which can be be found in Appendix~\ref{appendix:animations}).
  \begin{figure}
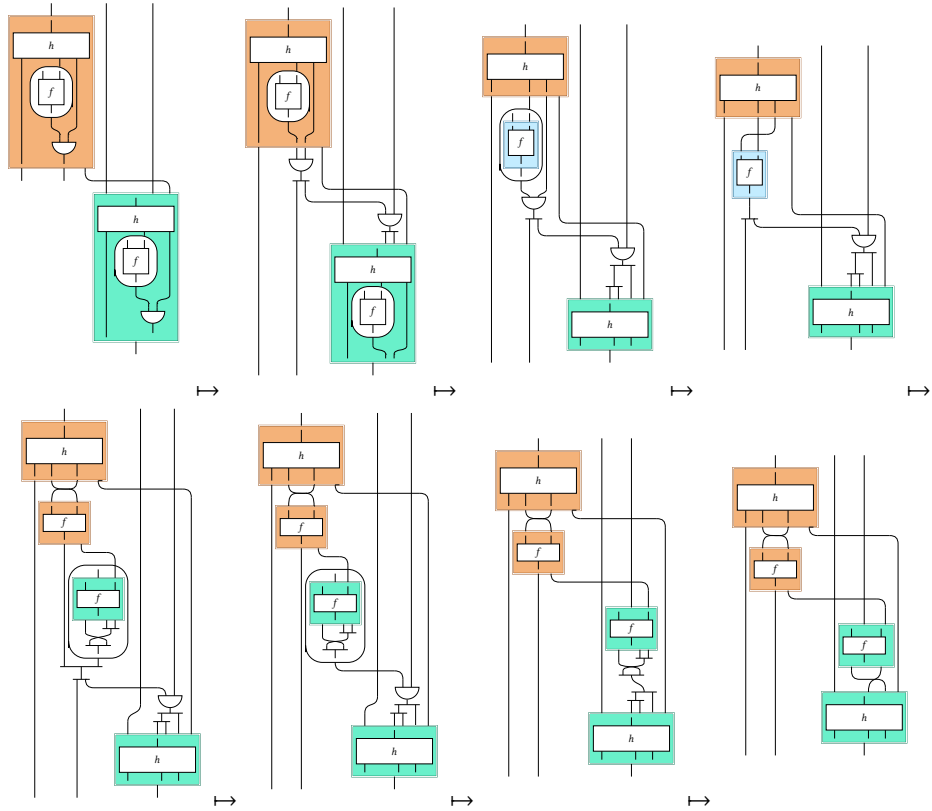

    \centering
    \scalebox{0.45}{\diagram{pics/ad/beta-local-proof-1}}
    $\rewrite$
    \scalebox{0.45}{\diagram{pics/ad/beta-local-proof-2}}
    $\rewrite$
    \scalebox{0.45}{\diagram{pics/ad/beta-local-proof-3}}
    $\rewrite$
    \scalebox{0.45}{\diagram{pics/ad/beta-local-proof-4}}
    $\rewrite$
    \scalebox{0.45}{\diagram{pics/ad/beta-local-proof-5}}
    $\rewrite$
    \scalebox{0.45}{\diagram{pics/ad/beta-local-proof-6}}
    $\rewrite$
    \scalebox{0.45}{\diagram{pics/ad/beta-local-proof-7}}
    $\rewrite$
    \scalebox{0.45}{\diagram{pics/ad/beta-local-proof-8}}
    \caption{Beta-soundness of AD\label{fig:ad-beta-proof}}
  \end{figure}
\end{proof}

\begin{lemma}
  For every diagram \opgeneric{f}{f} whose operands and results are all of a first-order type, there is a Beta-equivalent diagram
  \opgeneric{f'}{f-defunctionalized} that contains no instances of abstraction or evaluation and whose every node has all
  first-order inputs and outputs.
  \label{lem:defunctionalization}
\end{lemma}
\begin{proof}
  Since our graphical language is simply-typed, evaluation of the diagram \opgeneric{f}{f} is guaranteed to terminate,
  following an argument similar to the proofs of strong normalisation for the simply-typed $\lambda$-calculus (such
  as the one in \cite[Chapter 6]{girard1989proofs}), or for interaction nets (e.g. \cite{DBLP:journals/tcs/Mackie00}).
  Such a normal form cannot contain any redexes, and so any application or evaluation node must be connected to an operand 
  or a result wire, which cannot be the case as these have all first-order types.
\end{proof}

\begin{theorem}
  For every diagram \opgeneric{f}{f} whose operands and results are all first-order, Eqn.~\ref{eqn:first-order-soundness} holds.
  \begin{figure}[h]
    \centering
    \subcaptionbox{Adjoint application\label{eqn:first-order-soundness}}{
      \scalebox{.7}{\diagram{pics/ad/soundness-ad}} $=$ \scalebox{.6}{\diagram{pics/ad/soundness-grad}}
    }
    \subcaptionbox{Application after Beta reduction\label{eqn:first-order-soundness-step}}{
      \hspace{25pt}
      \scalebox{.7}{\diagram{pics/ad/ad-soundness-2}} $=$ \scalebox{.6}{\diagram{pics/ad/soundness-grad}}
      \hspace{25pt}
    }
    \renewcommand{\figurename}{Eqn.}
    \caption{First-order soundness for AD}
  \end{figure}
\end{theorem}
\begin{proof}
  Applying Lemma~\ref{lem:defunctionalization} and Lemma~\ref{lem:ad-beta}, it suffices to consider the case where
  $\opgeneric{f}{f}$ contains no instances of application or evaluation. Applying the rewrite rule in Fig.~\ref{fig:backprop-sd}
  and calculating gives the diagram in Eqn.~\ref{eqn:first-order-soundness-step}.
  The result then follows by induction on the foliation of the diagram $\opgeneric{f}{f}$. We show one case in full.

  \begin{center}
    \scalebox{.55}{\diagram{pics/ad/soundness-dup-grad-1}}
    =\scalebox{.55}{\diagram{pics/ad/soundness-dup-grad-2}}
    =\scalebox{.55}{\diagram{pics/ad/soundness-dup-grad-3}}
    =\scalebox{.55}{\diagram{pics/ad/soundness-dup-grad-4}}
    =\scalebox{.55}{\diagram{pics/ad/soundness-dup-grad-5}}=\\
    \scalebox{.55}{\diagram{pics/ad/soundness-dup-grad-6}}
    =\scalebox{.55}{\diagram{pics/ad/soundness-dup-grad-7}}
    =\scalebox{.55}{\diagram{pics/ad/soundness-dup-grad-8}}
    =\scalebox{.55}{\diagram{pics/ad/soundness-dup-grad-9}}
    =\scalebox{.55}{\diagram{pics/ad/soundness-dup-grad-10}}
  \end{center}
\end{proof}

\section{Hierarchical hypergraphs and rewriting}\label{sec:hypernet}

\renewcommand{\R}[0]{\mathcal{R}}

Even though string diagrams are convenient for mathematical reasoning, 
the actual implementation of string diagram rewriting poses a challenge. 
In order to perform a rewrite step, we need to find a match in a string diagram, but 
the presence of a redex may depend on which representation of the diagram we pick 
among ones that are equivalent according to the laws of symmetric monoidal categories~\cite{DBLP:conf/lics/BonchiGKSZ16}. 
A solution is identifying a data structure interpreting string diagrams
with the property that equivalent representations of a string diagram all have the same interpretation. 
For standard string diagrams in symmetric monoidal categories, 
such a data structure is provided by hypergraphs with interfaces, and 
the interpretation allows to model string diagram rewriting efficiently as \emph{double-pushout rewriting} of the corresponding hypergraphs~\cite{DBLP:conf/lics/BonchiGKSZ16}. 
In this section, we do something similar for our hierarchical hypergraphs: we devise a suitable combinatorial structure, called \emph{hypernets}, and 
show that two hierarchical hypergraphs are interpreted as the same hypernet whenever they are equivalent modulo the laws of symmetric monoidal categories (theorem~\ref{thm:def}). 
This allows us to conclude that string diagram rewriting for hierarchical hypergraphs can be `implemented' as double-pushout rewriting of hypernets (theorem~\ref{thm:rewrite}).

Hierarchical hypergraphs have been used before many times, see e.g.~\cite{DBLP:journals/cuza/BruniGL10,DBLP:journals/jcss/DrewesHP02,DBLP:journals/jcss/Palacz04}.
Our approach is broadly similar, but 
with enough subtle differences that it is necessary to give our own definitions.
For a more detailed comparison, see Sec.~\ref{ssec:related-rewriting}.

\subsection{Hierarchical hypergraphs and hypernets}

A hierarchical hypergraph is a labelled, directed hypergraph 
with a parent relationship which determines the hierarchical structure.
We fix sets of vertex and edge labels $\Sigma_V$ and $\Sigma_E$.
When comparing hierarchical hypergraphs to string diagrams,
$\Sigma_V$ should also be the set of base types in the string diagrams,
while $\Sigma_E$ should be the added operations.

\begin{definition}
A \textit{hierarchical hypergraph} is a tuple $(V, E, s, t, \ell_V, \ell_E, p_V, p_E)$
comprising a finite set of vertices $V$, a finite set of edges $E$,
source and target functions $s, t: E \to V^*$,
labelling functions $\ell_V: V \to \Sigma_V$ and $\ell_E: E \to \Sigma_E + 1$,
and parent functions $p_V: V \to E+1$ and $p_E: E \to E+1$.

While the source, target, and labelling functions are standard for labelled, directed hypergraphs,
we must add some conditions to the the parent functions.
First, we require an edge and any of its source and target vertices to have the same parent:
namely that $p_V(v) = p_E(e) = p_V(v')$ for all $v \in s(e)$ and $v' \in t(e)$ respectively.
Second, the parent relation must be acyclic,
so that repeatedly applying $p_E$ should eventually end up in the right summand of $E+1$.
More precisely, we assume for all $e \in E$ there is some $k \geq 1$ such that $(p_{E, \bot})^k(e) = \bot$ 
where $\bot$ is the element of 1 and
$p_{E, \bot}: E+1 \to E+1$ is the extension of $p_E$ adding $p_{E,\bot}(\bot) = \bot$.
\end{definition}

When the parent of a vertex or edge is the element $\bot$ from the right summand,
we say it is a \emph{outermost vertex} or \emph{outermost edge}.
If the label of an edge is $\bot$, we say (with some abuse) that it is \emph{unlabelled}.
When considering multiple hierarchical hypergraphs, 
we use subscripts to disambiguate these data.

We borrow terms from graph theory for hierarchical hypergraphs.
An important example is that we call a hierarchical hypergraph \emph{connected} when
for every pair of outermost vertices, there is a sequence of edges
(oriented either forward or reverse) joining the two vertices.

\begin{wrapfigure}{l}{0.35\textwidth}
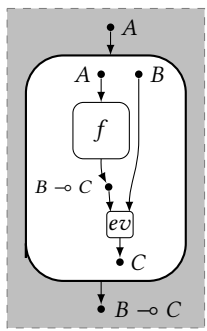

    \centering
    \diagram{pics/tikzit/hg-example}
    \caption{An example hierarchical hypergraph}
    \label{fig:hg-example}
\end{wrapfigure}

In every hierarchical hypergraph $\F$,
associated to every edge $\hat{e}$ is a subgraph consisting of edges $e$ (and vertices $v$)
satisfying $p_{E, \bot}^k(e) = \hat{e}$ (and $(p_{E, \bot})^j(p_V(v)) = \hat{e}$)
for some $k \geq 1$ (and $j \geq 0$).
We denote this subgraph $\F_{\hat e}$ and call it ``the inner hypergraph of $\hat{e}$''.
If a subgraph $\G$ of a hierarchical hypergraph $\F$ has the property that
$\F_e \subseteq \G$ for all $e \in E_\G$, we call $\G$ \emph{down-closed}.
When depicting a hierarchical hypergraph, 
we indicate the inner hypergraph of an edge by nesting the inner subgraph within its edge, 
like abstraction in hierarchical string diagrams.
An example can be seen in Fig.~\ref{fig:hg-example}.

We give an example hierarchical hypergraph in Fig.~\ref{fig:hg-example}.
This hypergraph has six vertices (the six black dots), 
which we name $v_1, \ldots, v_6$ from top to bottom then left to right.
These vertices are labelled as follows:
$\ell_V(v_1) = \ell_V(v_2) = A$, $\ell_V(v_3) = B$, $\ell_V(v_4) = \ell_V(v_6) = B\multimap C$,
and $\ell_V(v_5) = C$.
There are three edges: $e_1$ with label $f$, $e_2$ with label $ev$, and $e_3$ unlabelled but with an inner hypergraph.
The sources and targets of these edges are mostly clear (and mostly one-element lists),
except $s(e_2) = [v_4, v_3]$.
As mentioned above, $e_3$ is the parent edge for most of the graph, so
$p_E(e_3) = p_V(v_1) = p_V(v_6) = \bot$ and 
for all other edges and vertices the parent function returns $e_3$.

\begin{definition}
In a hierarchical hypergraph $\G$, a vertex is 
an \emph{input vertex} if it never occurs as a target,
an \emph{output vertex} if it never occurs as a source,
an \emph{interface} if it is either an input or an output vertex,
and an \emph{isolated vertex} if it is both an input and an output vertex.
\end{definition}

We  think of vertices of hierarchical hypergraphs as representing objects in a category
and edges representing morphisms from the product of the source objects to the product of the target objects.
However, hierarchical hypergraphs are generally much more expressive than string diagrams:
multiple edges can use the same vertex as a source or a target, 
and there could be cycles in the graph.
We will therefore be interested in a more restricted class of hypergraphs, which we call hypernets.

\begin{definition}
A \emph{hypernet} is a hierarchical hypergraph $\H$ with the following additional properties:
(1) acyclicity,
(2) all vertices occur as a source for at most one $e$ (and at most once in $s(e)$),
(3) all vertices occur as a target for at most one $e$ (and at most once in $t(e)$),
(4) there are specified total orderings on the input and output vertices, and
(5) $\ell_E(e) \neq \bot$ if and only if $\H_e$ is the empty hypergraph.
\end{definition}

For now, we return our focus to hierarchical hypergraphs 
and situate them in a category.

\begin{definition}
A \emph{morphism of hierarchical hypergraphs} $\phi: \F \to \G$ 
is a pair of functions $\phi = (\phi_V, \phi_E)$ with
$\phi_V: V_\F \to V_\G$ and $\phi_E: E_\F \to E_\G$.

These functions are required to respect the structure of the
hierarchical hypergraphs in the following senses:
\begin{enumerate}
\begin{multicols}{2}
    \item $\phi_V^* \circ s_\F = s_\G \circ \phi_E$
    \item $\phi_V^* \circ t_\F = t_\G \circ \phi_E$
    \item $\ell_{V, \G} = \phi_V \circ \ell_{V, \F}$
    \item $\ell_{E, \G} = \phi_E \circ \ell_{E, \F}$
    \item $(\phi_E \circ p_{V, \F})(v) = (p_{V, \G} \circ \phi_V)(v)$ if $p_V(v) \in E$
    \item $(\phi_E \circ p_{E, \F})(e) = (p_{E, \G} \circ \phi_E)(e)$ if $p_E(e) \in E$
\end{multicols}
\end{enumerate}
\end{definition}

Note that we do not require that outermost vertices and edges are sent to
outermost vertices and edges, due to the condition in (5) and (6).
If conditions (5) and (6) hold for all $v$ and $e$,
we say the morphism is \emph{strict}.

Hierarchical hypergraphs and the morphisms between them form a category.
This category clearly has finite coproducts given by disjoint union.
We investigate pushouts in this category in order to support double pushout rewriting.
When restricted to strict morphisms, all pushouts exist and can be computed as in $\Set$.
Unfortunately, the category of hierarchical hypergraphs with strict morphisms 
is not expressive enough for the rewriting tasks we require.

When allowing all hierarchical hypergraph morphisms,
the category does not have all pushouts or even pushouts along monos.
This is primarily due to ambiguities in the parents of outermost vertices and edges.
Two non-strict morphisms can embed a graph into two unmergeable parts of different graphs.
However, there are enough pushouts in this category
to support the double pushout structure we need.

In essence, given an arbitrary span $\L \xleftarrow{l} \K \xrightarrow{r} \R$
with the property that the outermost interfaces of $\L$ and $\R$ are isomorphic
together with a (monomorphic) matching $\L \xrightarrow{m} \G$ 
of the leftmost graph in the span in another hierarchical hypergraph,
the next few lemmas give conditions for a unique (up to isomorphism) $\G^-$ and $\H$
completing the following diagram, where all squares are pushouts:
\[\begin{CD}
\L @<[\iota_\L, \L]<< & I + \L @<I + l<< & I + \K @>I + r>> & I + \R @>[\iota_\R, \R]>> & \R \\
@VmVV & @V{n+\L}VV & @V{n+\K}VV & @V{n+\R}VV & @VpVV \\
\G @<\phi<< & \G^- + \L @<{I+l}<< & \G^- + \K @>{I+r}>> & \G^- + \R @>\psi>> & \H
\end{CD}\]

Here $I$ is a copy of the outermost interface vertices of $\L$ (and $\R$)
and $\iota_\L$ (resp.~$\iota_\R$) is the inclusion of these vertices in the graph.
That the inner two squares are the only pushout or pushout complement in this format is straightforward.
Lemma~\ref{lem:po-complement} gives the requirements on $m$ in the leftmost square to make it a pushout complement.
These conditions entail $n$ being a mono, which we will see in Lemma~\ref{lem:pushout} is an important critera to get the existence of a pushout in the rightmost square.

\begin{example}\label{ex:po}
We will illustrate the graph rewriting process with a running example.
For now, we just give an example of the input data we are expecting:
a span and a matching.
We start with a span corresponding to a particular instance of the Abs rule and
a matching of the left-hand side of this span in another hierarchical hypergraph.
\[\diagram{pics/tikzit/dpo/G}
  \xleftarrow{m}
  \diagram{pics/tikzit/dpo/L}
  \xleftarrow{l}
  \diagram{pics/tikzit/dpo/K}
  \xrightarrow{r}
  \diagram{pics/tikzit/dpo/R}\]
To avoid clutter, we omit the vertex labels and do not give a full description
of the morphisms other than to say they are the obvious map preserving edge labels.
The goal of this section is to formally describe how the copy of $\L$ in $\G$
is replaced with a copy of $\R$.
\end{example}

\begin{definition}
Suppose $I$ is a hierarchical hypergraph consisting of only isolated vertices,
and let $\phi: I \to \F$ and $\psi: I \to \G$ be morphisms.
We say $\phi$ and $\psi$ have \emph{complimentary images}
if 
\begin{enumerate}
  \item $p_V(\phi_V(i)) = p_V(\phi_V(j))$ for all $i, j \in V_I$,
  \item $p_V(\psi_V(i)) = p_V(\psi_V(j))$ for all $i, j \in V_I$, and
  \item either $\phi_V(i)$ never occurs as a source and $\psi_V(i)$ never occurs as a target
  or vice versa.
\end{enumerate}
\end{definition}

Though the following result does not completely characterize pushouts in this category,
it gives enough pushouts for us to construct the rightmost square in the diagram above.

\begin{lemma}\label{lem:pushout}
Suppose $I$ is a hypergraph of isolated vertices,
$\phi: I \to \F$ sends all vertices of $I$ to outermost vertices 
in the connected hypergraph $\F$,
and the images of the vertices under $\psi: I \to \G$ have a single common parent.
Then the span $\G + \F \xleftarrow{\psi + \id_\F} I + \F \xrightarrow{[\phi, \id_\F]} \F$
has a pushout.

If further $\phi$ and $\psi$ are monos and they have complimentary images,
then $\F$ and $\G$ being hypernets implies the pushout is as well.
\end{lemma}
\begin{proof}
The pushout can be formed by taking the disjoint union of $\F$ and $\G$,
then identifying the images of $I$ under the respective maps.
The parent of the outermost vertices and edges of $\F$ is defined to be
the common parent of the $\psi(i)$.
The remaining properties of hierarchical hypergraphs are straightforward.

That $\phi$ and $\psi$ are monos with complimentary images ensures that
when this identification occurs, equivalence classes of vertices have
at most two representatives (one from $\F$ and one from $\G$)
and that the resulting equivalence class is used as an input or an output
by at most one edge from either graph.
This makes the quotient is a hypernet.
\end{proof}

Next we establish the result constructing pushout complements in the leftmost square.

\begin{lemma}\label{lem:po-complement}
Suppose $I$ is a hypergraph of isolated vertices, 
$\phi: I \to \F$ is a bijection of vertices of $I$ with outermost interface vertices
in the connected graph $\F$.
Further suppose $m: \F \to \G$ is a monomorphism with the following properties:
(1) the image of $m$ is down-closed, and
(2) edges in outside the image of $m$ are incident only with
vertices outside the image of $m$ or vertices in $(m \circ \phi)(I)$.
Then there is a unique graph $\G^-$ such that
$(\G^- + \F, n + \F: I + \F \to \G^- + \F, \xi: \G^- + \F \to \G)$ is a
pushout complement to $(m, \phi)$.
Further, $n$ is a monomorphism.
\end{lemma}
\begin{proof}
Note that condition (2) is the dangling condition from double pushout rewriting,
and that $m$ is a monomorphism strengthens the identification condition.
Hence, it is not surprising to define 
$\G^- := \G \smallsetminus (m(\F) \smallsetminus (m\circ \phi)(I))$.
The wrinkle introduced by hierarchical hypergraphs is the hierarchy:
down-closedness (1) is required in order to remove (most of) the image of $m$ from $\G$
without introducing ambiguity in the parent relationships.
When an edge is in the image of $m$ in $\G$, this condition ensures
all of its children are also in the image, so we do not need to
redefine its parent after deletion.
\end{proof}

\begin{example}
With an understanding of how these pushouts and pushout complements are constructed, 
we now complete example~\ref{ex:po}.
The graphs formed in the pushout diagram are shown in Fig.~\ref{fig:ex-po}.
\begin{figure}[h]
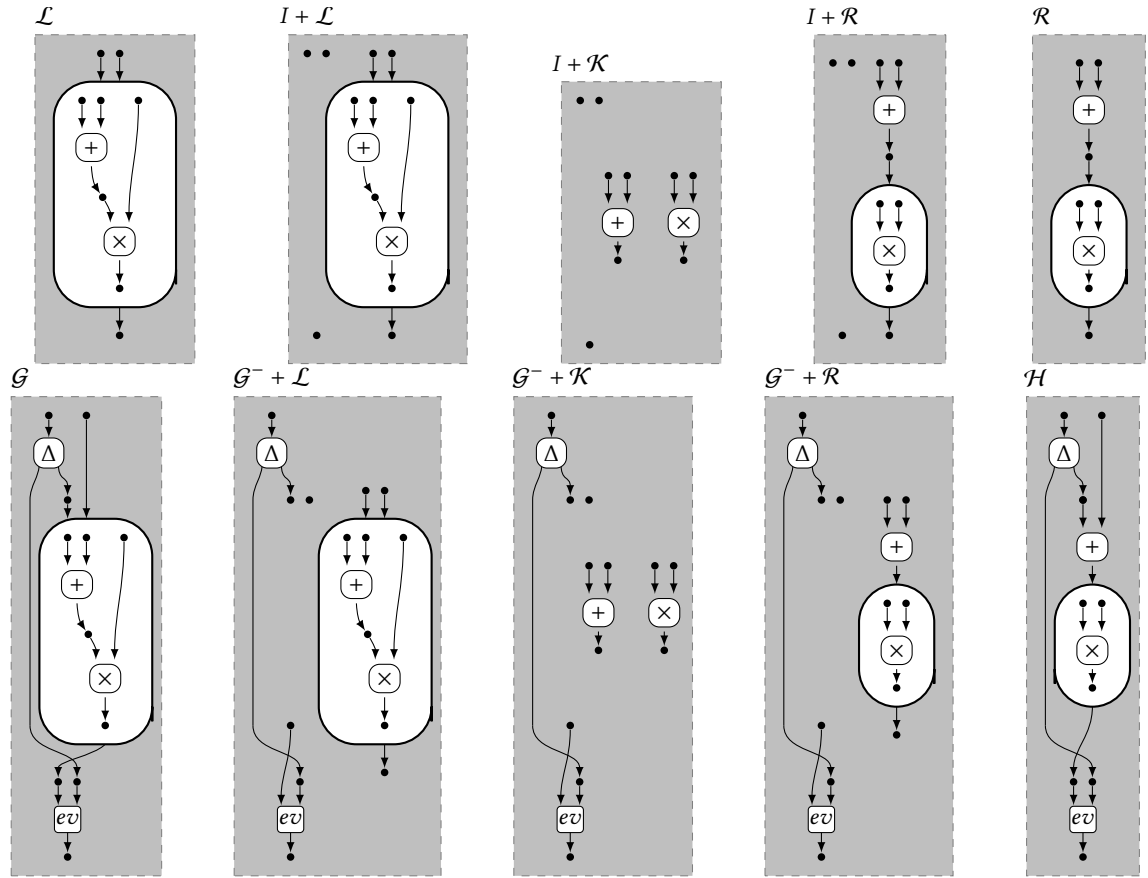

  \centering
  \scalebox{1}{\diagram{pics/tikzit/dpo/L}}
  \hfill
  \scalebox{1}{\diagram{pics/tikzit/dpo/I+L}}
  \hfill
  \scalebox{1}{\diagram{pics/tikzit/dpo/I+K}}
  \hfill
  \scalebox{1}{\diagram{pics/tikzit/dpo/I+R}}
  \hfill
  \diagram{pics/tikzit/dpo/R}\\

  \diagram{pics/tikzit/dpo/G}
  \hfill
  \diagram{pics/tikzit/dpo/G-+L}
  \hfill
  \diagram{pics/tikzit/dpo/G-+K}
  \hfill
  \diagram{pics/tikzit/dpo/G-+R}
  \hfill
  \diagram{pics/tikzit/dpo/H}
  \caption{Pushout rewriting example}\label{fig:ex-po}
\end{figure}

The leftmost pushout (complement) square excises the
matching of $\L$ from the graph in which it is embedded.
The next two squares replace $\L$ with $\R$ 
while the portion of the morphism from $I$ marks where the result should be reinserted in $\G^-$.
Finally, the rightmost pushout glues $\R$ back into $\G^-$.
\end{example}

\begin{remark}
This example illustrates one of the distinctive features of our approach to hierarchical hypergraph rewriting, namely the ability to send outermost vertices (and edges) to inner vertices (and edges). This is crucial, because both legs of the span require it. Equally crucial, but maybe less obvious, is that morphisms can sent outermost vertices (and edges) to images \emph{with different parents}, as seen in the right leg of the span. It is these novel properties that provide the required level of expressiveness we need to formulate our string diagram equations as graph rewrites.
\end{remark}

\subsection{Soundness and completeness}

Now that we understand how hierarchical hypergraphs can be rewritten,
we turn to the the connection between string diagrams and hypernets.
Suppose we fix a set of basic types $\Sigma_V$ and a set of operations $\Sigma_E$
from which the string diagrams of section 2 are built.
From these basic types, generate types according to the rules of~\ref{sec:types}.
We restrict vertex labels for our hypernets to these atomic types.
Just as with string diagrams, there is a notion of well-typedness for hypernets.

\begin{definition}
  A \emph{well-typed} hypernet $\H$ satisfies the following properties.
  (1) If $\ell_E(e)$ is an operation, then the types of $s(e)$ match the input types of this operation
  and similarly the types of $t(e)$ match the output types of $\ell_E(e)$.
  (2) If $\ell_E(e) = \bot$, let $L_i$ be the types of the list of input interfaces in $\H_e$ in the interface order given on the hypernet.
  Similarly let $L_o$ be the list of types on the output interfaces of $\H_e$.
  Then there is a list of types $L_a$ such that $s(e) :: L_a = L_i$ and $t(e) = (L_a \multimap L_o)$.
\end{definition}

As we will see, though property (4) of hypernets 
(a specified total ordering on sources and targets)
is not so important in rewriting, it is necessary to establish a typing for hypernets.
To reflect this, we depict the interface ordering graphically by 
drawing a copy of the interface in the specified ordering left-to-right in a blue box 
and the correspondence between the ordered interface and the hypergraph.

\begin{figure}[ht]
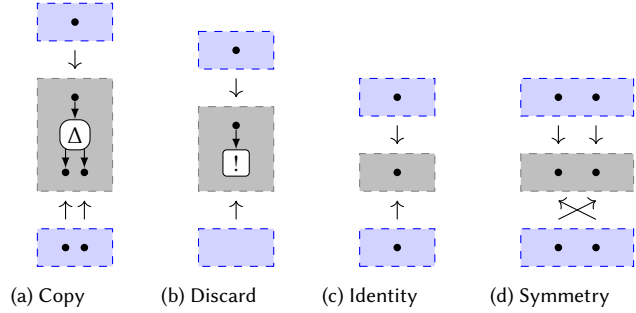

\subcaptionbox{Copy\label{fig:copy-hg}}{
\quad \diagram{pics/tikzit/components/copy-hg}}\quad
\subcaptionbox{Discard\label{fig:discard-hg}}{
\quad \diagram{pics/tikzit/components/discard-hg}}\quad
\subcaptionbox{Identity\label{fig:id-hg}}{
\quad \diagram{pics/tikzit/components/id-hg}}\quad
\subcaptionbox{Symmetry\label{fig:swap-hg}}{
\quad \diagram{pics/tikzit/components/swap-hg}}\quad
\caption{Interpretation of string diagrams in hypernets}
\label{fig:hg-interpretation}
\end{figure}
Every string diagram has an interpretation as a hypernet with these labels, 
inductively defined based on a(ny) decomposition of the string diagram.
For base cases, contraction, weakening, evaluation, and any of the basic operations 
are interpreted as a single hyperedge with the corresponding label from $\Sigma_E$.
Identity and unary swap are slightly subtle: 
in the hypernet representation they do not require an edge.
Rather, they are graphs of isolated vertices 
with different input and output orderings,
as show in Figures~\ref{fig:id-hg} and \ref{fig:swap-hg}.
For induction steps, compositions and abstraction are as expected,
with abstraction taking advantage of the hierarchical structure.
We write $\seval S$ for the interpretation of a string diagram $S$ under this scheme.

Note that this interepretation scheme absorbs 
the equations required of a symmetric monoidal category.
Sequentially composing an identity hypernet with any other hypernet
(that can be composed with that identity)
results in a hypernet isomorphic to the given hypernet.
Similarly, associativity of compositions and tensoring with the unit
also yield isomorphic hypernets.
Finally, the two sides of the equations for
naturality and idempotency of symmetry, when interpreted under this scheme,
also result in isomorphic hypernets.

In the other direction, we can show that every hypernet arises as
the interpretation of a string diagram, and further that
all string diagrams which have a given hypernet as their interpretation
are equivalent modulo the equations of symmetric monoidal categories.
Before we do this, it is useful to establish a result about hypernets 
similar to the foliation decomposition of lemma~\ref{lem:max-seq-foliation}.

\begin{lemma}\label{lem:before-during-after}
Let $\H$ be a hypernet, and $\G \subseteq \H$ be a connected, down-closed, outermost-level subnet.
Then there is a decomposition of $\H$ into the sequential composition of
(1) a hypernet, (2) the parallel composition of $\G$ with some identity hypernets,
and (3) another hypernet.
\end{lemma}
\begin{proof}
Topologically sort the hypernet. 
The edges after the last edge of $\G$ form (3).
Edges before the last of $\G$ but not including any edges from $\G$ form (1).
Finally, any edges of $\G$, together with new vertices (identity hypernets) 
for any outputs of (1) not used in $\G$ form (2).
\end{proof}

Note that this decomposition is certainly not unique.

\begin{theorem}[Hierarchical Definability]\label{thm:def}
Let $H$ be well-typed hypernet. 
There exists a string diagram $S$ such that $\seval S = H$. 
If $S'$ is any other string diagram with the property that $\seval {S'} = H$,
then $S =_{SMC} S'$.
\end{theorem}
\begin{proof}
For existence, we induct on the number of edges of the hypernet.
If the hypernet has no edges, the output ordering is some permutation of the input ordering.
Since all permutations are generated by the set of adjacent transpositions,
there is a combination of unary swaps which has this hypernet as its interpretation.

If the hypernet has at least one edge, it has a outermost-level edge. 
Let $e$ be an arbitrary outermost-level edge, and $\G$ be the hypernet inside $e$ (if there is one).
By lemma~\ref{lem:before-during-after}, we can decompose this hypernet into three pieces.
Hypernets (1) and (3) do not contain $e$,
so by the induction hypothesis they have a string diagram representation.
If (2) has a basic label, it is the interpretation of the corresponding symbol in the signature.
If (2) does not have a basic label (and thus has an interior hypernet),
the induction hypothesis again tells us the interior is the interpretation of a string diagram.
Then $e$ is the interpretation of the abstraction of that string diagram.
The sequential composition of these three string diagrams is then a string diagram represented by $H$.
\end{proof}

We denote the string diagram (up to symmetric monoidal axioms) 
of a hypernet $H$ as $\llparenthesis H\rrparenthesis$.

\begin{lemma}
Let $\H$ be a hypernet. 
For every connected subnet $\G \subseteq \H$ which contains all its descendents
(as in assumption (1) of Lem.~\ref{lem:po-complement}),
there is a decomposition of the string diagram $\llparenthesis \H \rrparenthesis$
which includes the string diagram $\llparenthesis \G \rrparenthesis$.
\end{lemma}
\begin{proof}
Combine Lem.~\ref{lem:before-during-after} and Thm.~\ref{thm:def}.
\end{proof}

Axioms of the Cartesian structure are not ``baked in'' to the hypernet structure
in the way the properties of symmetric monoidal structures are.
Instead, these equations are modeled by bidirectional rewriting rules.

\begin{theorem}\label{thm:rewrite}
Every equation of Cartesian closed categories has a
corresponding bidirectional rewrite rule.
That is, for every equation $L = R$ of Cartesian closed categories (expressed as string diagrams)
there is a rewrite rule $\seval L \leftrightarrow \seval R$
such that applying this rule in a hypernet
\end{theorem}
\begin{proof}
It is a straightforward, but tedious, exercise to find spans in the style of example~\ref{ex:po}
relating the left- and right-hand sides of each required equation, and
checking that they satisfy the conditions of lemmas~\ref{lem:pushout} and \ref{lem:po-complement}.
\end{proof}

\section{Related work}

\subsection{String diagrams}

Cartesian closed categories have been thoroughly studied in the context of logic and type theory, because of the well-known correspondence of their internal language with $\lambda$-calculus and intuitionistic logic~\cite{sorensen2006lectures}. 
The \textit{linear} version of this triad involves monoidal rather than cartesian categories, but also proof nets, and linear logic, as indicated already in the original paper~\cite{DBLP:journals/tcs/Girard87}.
\cite{DBLP:conf/csl/Mellies06} provides the foundation on which we build our language of string diagrams, noting that all the basic ingredients are already there.  

The route of using an enhancement of the monoidal closed structure with additional properties to control sharing is fruitful and has been employed many times. 
For example it is found in~\cite{DBLP:journals/jlp/BonchiSZ18}, where the manipulation of variables endowed with algebraic theories is modeled as a cartesian structure on the top of a linear structure, or in~\cite{DBLP:conf/popl/Ghica07} to specify multiplexers and demultiplexers in high-level synthesis. 

To the best of our knowledge, we provide the first fully specified string diagrammatic language for cartesian closed categories generated as a graphical syntax quotiented by equations. 
Our approach shares similarities with the formalisms of sharing graphs for describing $\lambda$-calculus computations \cite{DBLP:conf/popl/Lamping90}. 
The main difference is that string diagrams, albeit graphical in appearance, can be manipulated as a syntax, whereas sharing graphs are usually studied as combinatorial objects. 
Unlike syntax, reasoning about graphs algebraically requires a higher degree of technical sophistication~\cite{DBLP:journals/tcs/Guerrini99}. 
Finally, sharing graphs are typically used to study low-level computational models for functional languages, in particular quantitative models~\cite{DBLP:journals/lmcs/MuroyaG19}, whereas our approach is more focussed on equational reasoning and rewriting, and does not have the ambition of investigating the resources employed during computation. 

Monoidal closed categories extend not only to cartesian closed categories, but also to $\star$-autonomous categories. 
This second variation is very much relevant to the study of multiplicative linear logic and it has been extensively studied in terms of proof nets. 
Our graphical calculus is essentially different from proof nets. 
The grammar of generating morphism does not stem from a sequent calculus, and we capture the intended semantics via equations rather than a correctness criterion.
But the connection might be made precise relying on the existing translations between proof nets and string diagrams~\cite{hughes,shulman}.
This is not the only possible extension. 
For example, another direction of extending monoidal categories is to traced monoidal categories~\cite{DBLP:conf/tlca/Hasegawa97}, which has interesting applications to modeling circuits with feedback~\cite{DBLP:conf/csl/GhicaJL17}.
Finally, a different style of hierarchical string diagrams appear in the literature to represent universal properties graphically such as Kan extensions~\cite{DBLP:conf/mpc/Hinze12} and free monads~\cite{DBLP:conf/icfp/PirogW16}. 

The only other proposal for a string-diagram language for monoidal closed categories which we are aware of is that of~\cite{baez2010physics}. 
We found a great deal of inspiration in their proposal, but in the process of fully working out the equational properties and the combinatorial structure we felt compelled to deviate somewhat from \textit{loc. cit.}.
To keep the language of types as simple as possible and as \textit{strict} as possible they propose an intriguing graphical innovation, a so-called \textit{clasp} operator on stems.
The exponential type is represented, using the clasp, as~\includegraphics[scale=.4]{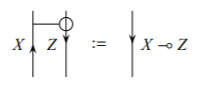}. 
Much like in our own language, a \textit{bubble} is used to represent currying. 
A simple example which uses both these graphical devices is the \textit{name} of a function $f:X\to Y$, represented as~\includegraphics[scale=.5]{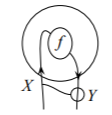}. 
We found it difficult to work out some of the unspecified details in particular for the \textit{clasp}, e.g. how it can be used to represent higher-order objects (e.g. $(A\multimap B)\multimap C$). 
But, more fundamentally, it was unclear to us what status we can give the clasp both as a syntactic and as a combinatorial object. 
To conclude, the approach in \textit{loc. cit.} is ambitious and innovative and, if the details were figured out, potentially more elegant in that it preserves an appealing visual parallel between monoidal closed and compact closed structures. 
However, filling in the missing details proved to be too challenging. 

\subsection{Rewriting of hierarchical graphs}\label{ssec:related-rewriting}

The notion of hierarchical hypergraph used in this paper is closely inspired by, and a formalisation of the graphs used in~\cite{spartan}. 

Although there is no consensus on a standard definition of hierarchical graphs, the various approaches to rewriting on these structures~\cite{DBLP:journals/cuza/BruniGL10,DBLP:journals/jcss/DrewesHP02,DBLP:journals/jcss/Palacz04} give slight variations on the idea of graphs containing other graphs and notions of morphism between them.
Some of the variations are minor: our hierarchical hypergraphs are directed, but some works do not make this choice~\cite{DBLP:journals/jcss/DrewesHP02}.
Other differences are much more stark.
Sometimes edges are permitted to connect vertices with different parents, as in~\cite{DBLP:journals/jcss/Palacz04}, sometimes this is prohibited (as it is here), and sometimes it is possible with the aid of an explicit renaming function, as in~\cite{DBLP:journals/cuza/BruniGL10}.
Some approaches consider only ``strict morphisms'' sending items in the outermost level to the outermost level, but others consider a larger class where this need not hold.
Due to the subtle but technically significant differences between our requirements and the properties of previous works, it was not possible to reuse previous work wholesale, and we found it necessary to introduce our own variation.

The formal correspondence between monoidal closed categories and hierarchical hypergraphs lies in a tradition of analogous results relating string diagram rewriting and double-pushout hypergraph rewriting, see~\cite{DBLP:conf/lics/BonchiGKSZ16}. 
To the best of our knowledge, such correspondence has not been spelled out in the way presented in our work, although the idea of linking the exponential structure of closed categories with the hierarchy structure of hierarchical hypergraphs may be found in~\cite{DBLP:journals/entcs/CocciaGM02}. 
Although it does not uses string diagrams or other categorical tools, the algebraic specification language for hierarchical graphs studied in \cite{DBLP:conf/tgc/BruniGL10} is aiming towards similar goals. 

\subsection{Syntax as a graph-like data structure}

Representing intermediate stages of the compiler as graphs is a long-established practice in compiler design and engineering. 
Graphs are an \emph{efficient} syntactic representation which are recognised as a better target for optimisation and analysis than raw text. 
In its simplest incarnation the graph representation of terms is just a \emph{abstract syntax tree}, but more sophisticated representations were increasingly used~\cite{DBLP:conf/irep/ClickP95}, sometimes leading to specific and novel optimisation techniques~\cite{DBLP:conf/pldi/NandiWAWDGT20}.

The use of graph-like representation outside of compiler engineering has a lot of untapped potential, as advocated by some~\cite{DBLP:journals/corr/abs-2102-02363}. 
This is not entirely new, for example interaction nets are a graph-like semantics of higher-order computation~\cite{DBLP:conf/popl/Lafont90}, but they are specified at a fairly high level of informality which string diagrams and hypernets make fully formal in two different ways. 

Although not presented explicitly as a string diagram language, the treatment of closures in~\cite{DBLP:journals/entcs/SchweimeierJ99} is  related in methodology to our work, although the use of \emph{partially-traced partially-closed pretmonoidal categories} as the categorical setting, in order to accommodate for effects, is significantly different than our cartesian closed categorical language. 

Finally, another related line of work which we found inspirational is the use of graph-like languages inspired by proof nets to bridge the gap between syntax and abstract machines, in order to provide a quantitative analysis of reduction strategies for the lambda calculus~\cite{DBLP:journals/tcs/Accattoli15}. 

\subsection{Automatic differentiation}

Our AD algorithm is an adaptation of~\cite{DBLP:journals/toplas/PearlmutterS08}.
Beyond the presentation based on string diagrams, the main differences are that our algorithm applies to simply-typed, recursion-free code and it acts as a source-to-source transformation, lacking the reflection features that enabled higher-order differentiation in the original work.
We chose to focus on this algorithm for a few reasons: first, reverse-mode AD is both more immediately useful (see~\cite{DBLP:journals/jmlr/BaydinPRS17} for a comparison of both approaches) and harder to implement and prove correct than forward-mode AD. Simple forward-mode AD algorithms based on operator overloading~\cite{DBLP:conf/icfp/Karczmarczuk98,DBLP:conf/popl/PearlmutterS07} capable of handling higher-order functions predate~\cite{DBLP:journals/toplas/PearlmutterS08}.
Second, it is to our knowledge the first published algorithm for performing reverse-mode AD on higher-order code\footnote{An earlier algorithm appears in \cite{karczmarczuk2000adjoint}, however it is argued in~\cite{DBLP:journals/toplas/PearlmutterS08} that this algorithm results in different computation graphs -- and worse asymptotic complexity -- than `traditional' reverse-mode AD},
it forms the basis of a number of efficient implementations~\cite{DBLP:journals/corr/BaydinPS16,DBLP:journals/corr/SiskindP16a} and does not require more complex features, unlike e.g.~\cite{DBLP:journals/pacmpl/WangZDWER19} which makes use of mutable state and continuations or~\cite{DBLP:journals/pacmpl/BrunelMP20} which relies on a limited form of continuations to encode dual spaces.

A wave of recent research has also tackled the issues of correctness in automatic differentiation. 
Notably,~\cite{DBLP:journals/pacmpl/BrunelMP20} and~\cite{DBLP:conf/esop/Vakar21} provide correct reverse-mode AD algorithms capable of handling closures. Unlike the first work, however, our algorithm is purely functional and, while the second one can correctly differentiate terms with higher-order inputs and outputs, it achieves so by using a more expensive representation of tangents of function spaces.
The main contribution of our approach, however, is the simplicity of the involved proofs thanks to our diagrammatic notation which we believe improves on the readability of the original paper~\cite{DBLP:journals/toplas/PearlmutterS08} and the denser proofs in newer literature~\cite{DBLP:journals/pacmpl/BrunelMP20,DBLP:conf/fossacs/HuotSV20,DBLP:conf/esop/Vakar21}.

\section{Conclusion and further work}
In this paper we have presented a recipe for provably correct reverse automatic differentiation built around a new, hierarchical, calculus of string diagrams. 
As we have seen, the string diagram presentation simplifies much of a bookkeeping of variables which a term calculus would require, which, we believe, makes a complicated algorithm more readable. 
More importantly, the new perspective offered by string diagrams and in particular the presentation of terms as \emph{foliations} opens the door for new and useful proof techniques. 
Finally, the combinatorial representation of string diagrams as hypergraphs makes it possible to formulate automatic differentiation using the established language of DPO graph rewriting.

In this paper we have not discussed implementation matters, yet these are of the essence. 
This algorithm is a practical one and it can be incorporated into real-life compilers for real-life programming languages. 
We surmise that the new an improved perspective on AD that string diagrams offers will help handle other challenging features of real-life languages, such as effects and, in particular, the crucial role that closures play. 
This is work is ongoing. 

\begin{acks}                            
  This material is based upon work supported by the
  Engineering and Physical Sciences Research Council (UK) under Grant
  EP/V001612/1 (Nominal String Diagrams).
\end{acks}

\bibliography{arxiv}

\newpage




\end{document}